\documentclass[10pt]{article}

\usepackage{amsmath,amssymb,amsthm,comment}
\usepackage{times}
\usepackage{fullpage}

\def\keywords{\vspace{.5em}
{\textit{Keywords}:\,\relax%
}}

\usepackage{tikz}
\usetikzlibrary{shapes,calc,arrows,trees,positioning,through,intersections,fadings}

\usepackage{graphicx}
\usepackage{color}
\usepackage{soul}
\definecolor{lightblue}{rgb}{.60,.60,1}
\definecolor{lightred}{rgb}{1, .60, 0.60}

\newcommand{\poly}{{\rm poly}}
\newcommand{\ie}{{\mbox{i.e.}}}
\def\zo{\{0,1\}}
\def\mapping{\rightarrow}
\newcommand{\prob}{{\rm Prob}}

\newcommand{\ci}{\mathrm{ComInf}}

\newcommand{\pxy}{p_{x \mid y}}

\newcommand{\xa}{x_A}
\newcommand{\xb}{x_B}
\newcommand{\xc}{x_C}
\newcommand{\ra}{r_A}
\newcommand{\rb}{r_B}
\newcommand{\rc}{r_C}
\newcommand{\pa}{p_A}
\newcommand{\pb}{p_B}
\newcommand{\pc}{p_C}
\newcommand{\xxa}{x_1}
\newcommand{\xxb}{x_2}
\newcommand{\xxc}{x_3}

\newcommand{\na}{n_1}
\newcommand{\nb}{n_2}
\newcommand{\nc}{n_3}
\newcommand{\leqp}{\leq^{+}}
\newcommand{\geqp}{\geq^{+}}
\newcommand{\eqp}{=^{+}}
\newcommand{\co}{\mathrm{CO}}

\newcommand{\tx}{\tilde{t}}
\newcommand{\tz}{\tilde{z}}

\newtheorem{theorem}{Theorem}[section]

\newtheorem{lemma}[theorem]{Lemma}
\newtheorem{claim}[theorem]{Claim}
\newtheorem{proposition}[theorem]{Proposition}

\newtheorem{definition}[theorem]{Definition}

\theoremstyle{remark}
\newtheorem{openq}{Open Question}
\newtheorem{remark}{Remark}

\begin{document}

\title{An operational characterization of mutual information in algorithmic information theory\thanks{%
A preliminary version of this work has been presented at 45thInternational Colloquium on Automata, Languages, and Programming (ICALP), Prague, July 10-13, 2018.}}

\author{ 
{Andrei Romashchenko\/}
\thanks{LIRMM, University of Montpellier  \& CNRS; on leave from IITP RAS.
Email:~\texttt{andrei.romashchenko@lirmm.fr};   %
 supported in part by  supported in part by the ANR through grant RaCAF ANR-15-CE40-0016-01. }
{\quad Marius Zimand\/}
\thanks{  Department of Computer and Information Sciences, Towson University,
Baltimore, MD. \texttt{http://orion.towson.edu/\~{ }mzimand/}; %
supported in part by the National Science Foundation through grant 1811729.}
}

\maketitle

\begin{abstract}
We show that the mutual information, in the sense of Kolmogorov complexity, of any pair of strings
$x$ and $y$ is equal, up to logarithmic precision, to the length of the longest shared secret key that
two parties, one having $x$ and the complexity profile of the pair and the other one having $y$ and the
complexity profile of the pair, can establish via a probabilistic protocol with interaction on a public channel.
For $\ell > 2$, the longest shared secret that can be established from a tuple of strings $(x_1, . . . , x_\ell)$ by  $\ell$
parties, each one having one component of the tuple and the complexity profile of the tuple, is equal, up
to logarithmic precision, to the complexity of the tuple minus the minimum communication necessary
for distributing the tuple to all parties. We establish the communication complexity of secret key agreement protocols  that produce a secret key of
maximal length, for protocols with public randomness.  We also show that if the communication complexity drops below the established threshold, then only very short secret keys can be obtained.
\end{abstract}


\section{Introduction}

Mutual information is a concept of central importance in  both information theory (IT) and algorithmic information theory (AIT), also known as Kolmogorov complexity.  We show an interpretation  of mutual information in AIT, which links it to  a basic notion in cryptography.  Even though a similar interpretation was known in the IT framework, an operational characterization of mutual information in AIT has been elusive till now.

To present our result, let us  consider  two strings $x$ and $y$.  It is common to draw a Venn-like diagram such as the one in Figure~\ref{f:figone} to visualize the information relations between them. 
 \begin{figure}[h]
 \centering{
\begin{tikzpicture}[shorten >=1pt,scale=0.15,
lNode/.style={fill=mygreen, circle},
rNode/.style={fill=myblue, circle}
]

\coordinate (x) at (0,0);
\coordinate (b) at (7,-7);
\coordinate (a) at (7,9);
\coordinate (y) at (12,0);
\coordinate (lowx) at (-0.5,-7.5);
\coordinate (lowxx) at (-3,-12);
\coordinate (lowy) at (13, -7.0);
\coordinate (lowyy) at (16, -12);
\coordinate(mida1) at (3,6);
\coordinate(mida2) at (11,6);
\coordinate (lowb) at (7,-10);
\coordinate(midb) at (7,-5);
\node(X) at (x) {$C(x)$};
\node(B) at (b)  {};
\node(Y) at (y) {\quad$C(y)$};

\node(c1) at (x) [draw, thick, fill=lightgray, circle through=(b), opacity=0.5] {};
\node(c2) at (y) [draw, dotted, thick, fill=lightgray, circle through=(b), opacity=0.5] {};
\node[anchor=north] (c3) at (lowxx.-90) {$C(x \mid y)$};
\node[anchor=north] (c4) at (lowyy.-90) {$C(y \mid x)$};
\node[anchor=south] (c5) at (a.90){$C(x,y)$};
\node[anchor=north] (c6) at (lowb.-90){$I(x:y)$};
\draw[->] (lowxx) -- (lowx);
\draw[->] (lowyy) -- (lowy);
\draw[->] (lowb)--(midb);
\draw[->] (a)--(mida1);
\draw[->] (a)--(mida2);
\end{tikzpicture}
}
\caption{Two strings $x$ and $y$, and their information.  There are six regions that we distinguish: (1) The left solid circle represents the information in $x$, as given by its Kolmogorov complexity, denoted $C(x)$; (2) The right dotted circle represents the information in $y$, denoted $C(y)$;  (3) The entire grey region (the two circles taken together) represents the information in $x$ and $y$, denoted $C(x,y)$; (4) The light-grey region in the first circle represents the information in $x$ conditioned by $y$, denoted $C(x \mid y)$; (5)  The light-grey region in the second circle represents the information of $y$ conditioned by $x$, denoted $C(y \mid x)$; and (6) the dark-grey region in the middle represents the mutual information of $x$ and $y$, denoted $I(x:y)$.   }
\label{f:figone}
\end{figure}
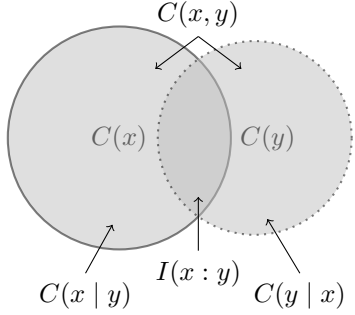
As explained in the figure legend there are six important regions. The regions $(1)$ to $(5)$ have a clear operational meaning. For instance, $C(x)$ is the length of a shortest program that prints $x$,  $C(x \mid y)$ is the length of a shortest program that prints $x$ when $y$ is given to it, and so on.   On the other hand, the mutual information $I(x:y)$ from region  $(6)$ is defined by a formula: $I(x:y) = C(x)+C(y) - C(x,y)$.  Intuitively, it is the information shared by $x$ and $y$.  But, is there an operational interpretation of  the mutual information? As mentioned above, we give a positive answer:
The mutual information of $x$ and $y$  is essentially equal to the length of a longest shared secret  key that two parties, one having $x$ and the other one having $y$, and both parties also possessing the complexity profile of the two strings, can establish via a probabilistic protocol.

\label{line-point-example}
The following  simple example illustrates the above concepts. Suppose that Alice and Bob want to agree on a common secret key. If they could meet face-to-face, they could just generate such a key by, say, flipping a coin. Unfortunately,  they cannot meet in person and what makes the situation really troublesome is that they can only communicate through a public channel.  There is, however,  a gleam of hope because Alice knows a random line  $x$ in the affine plane over the finite field with $2^n$ elements, and Bob knows a random point $y$ on this line.      The line $x$ is specified by the slope $a$ and the intercept $b$ and the point $y$ by its two coordinates $c$ and $d$. Therefore each of $x$ and $y$ has $2n$ bits of information, but, because of the geometrical correlation, together they  have $3n$ bits of information. Thus, in principle, Alice and Bob share $n$ bits.  Can they use them to obtain  a common secret key? 
 The answer is yes: Alice sends $a$ to Bob, Bob, knowing that his point is on the line, finds $x$, and now they can use $b$ as the secret key, because the adversary has only seen $a$, and $a$ and $b$ are independent.

It may appear that the geometrical relation between $x$ and $y$ is crucial for the above solution. In fact it is just a red herring and Alice and Bob can agree on a common secret key  in a very general setting.
To describe it, we consider the scenario in which Alice has a random string $x$ and Bob has a random string $y$. If $x=y$, then Alice and Bob can use their common string as a secret key in an encryption scheme (such as the one-time pad) and achieve perfect information-theoretical security.  What happens if $x$ and $y$ are not equal but only correlated? Somewhat surprisingly, for many interpretations of ``correlated,"  they can still agree on a shared secret key via interaction on a public channel (for instances of this assertion, see~\cite{leu:t:secret,ben-bra-rob:j:privacyamplification,mau:j:seckey,ahl-csi:j:seckeyone}).
In this paper, we look at  this phenomenon  using the very general framework of algorithmic information theory to measure the correlation of strings. 

We proceed with a description of our results.

\subsection{Our contributions}
\label{s:contrib}

{\bf Characterization of mutual information.} In a secret key agreement protocol, Alice and Bob, on input $x$ and respectively $y$, exchange messages and compute a common string that is random conditioned by the transcript of the protocol. Such a string is said to be a \emph{shared secret key}.
Unless specified otherwise, we use protocols having the following features:

(1) We assume that Alice and Bob also know how their $x$ and $y$ are correlated. In our setting this means that Alice and Bob know the complexity 
profile of $x$ and $y$, which, by definition, is the tuple $(C(x), C(y), C(x,y))$. 
\smallskip

(2) The protocols are effective and randomized, meaning that Alice and Bob use probabilistic algorithms to compute their messages. 
\smallskip

\begin{theorem}[Main Result, informal statement]
\label{t:main}
\begin{enumerate}
\item There is a secret key agreement protocol that, for every n-bit strings x and y, allows Alice and Bob to compute
with high probability a shared secret key  of length equal to the mutual information of x and y (up to
an $O(\log n)$ additive term). 
\item No protocol  can produce a
longer shared secret key (up to an $O(\log n)$ additive term).
\end{enumerate}
\end{theorem}
\medskip

{\bf Secret key agreement for  three or more parties.}  
Mutual information is only defined for two strings, but secret key agreement can be explored for the case of more  strings. Let us consider again an example. 
Suppose that each of Alice, Bob, and Charlie have a point in the affine plane over the finite field with $2^n$ elements, and that the three points, which we call $A, B, C$, are collinear. Thus, each party has $2n$ bits of information, but together they have $5n$ bits of information, because given two points, the third one can be  described with $n$ bits. The parties want to establish a common secret key, but they can only communicate by broadcasting messages over a public channel. They can proceed as follows. Alice will broadcast a string $\pa$, Bob a string $\pb$, and Charlie a string $\pc$, such that each party using his/her point and the received information will reconstruct the three collinear points $A, B, C$.  A protocol that achieves this is called an \emph{omniscience protocol} because it spreads to everyone the information possessed at the beginning individually by each party. In the next step, each party will compress the $5n$ bits, comprising the three points, to a string that is random given $\pa, \pb, \pc$. The compressed string is the common secret key. We will see that up to logarithmic precision it has length $5n - (|\pa| + |\pb| + |\pc|)$. Assuming we know how to do the omniscience protocol and the compression step, this protocol produces a common secret key of length $5n - \co(A,B,C)$, where $\co(A,B,C)$ is the minimum communication for the omniscience task for the points $A,B,C$. In our example, it is clear that each one of $\pa, \pb, \pc$ must be at least $n$ bits long, and that any two of these strings must contain together at least $3n$ bits. Using some recent results from the reference~\cite{zim:c:kolmslepianwolf}, it can be shown that any numbers satisfying these constraints can be used for the omniscience task. It follows that the smallest communication for omniscience is achieved when $|\pa| = |\pb| = |\pc| = 1.5n$, and thus the key has $5n -4.5n= 0.5n$ bits. (Warning: in the entire discussion we  ignore additive terms of size $O(\log n)$.) 
We show that this holds in general.  If $\ell$ parties have, respectively,  one component of a tuple $(x_1, \ldots, x_\ell)$ of $n$-bit strings, then up to $O(\log n)$ precision, they can produce a common secret key of length $C(x_1, \ldots, x_\ell) - \co(x_1, \ldots, x_\ell)$, where $\co(x_1, \ldots, x_\ell)$ is the minimum communication for the omniscience task (which, as we show in the paper, depends only on the complexity profile of the tuple). The protocol that produces such a key is probabilistic, and, as was the case for two strings, assumes that each party $i$ has at the beginning of the protocol besides its input string $x_i$ also the complexity profile of the entire tuple $(x_1, \ldots, x_\ell)$.
We also show a  matching (up to $O(\log n)$) upper bound, which holds for any number of rounds in the protocol.

\medskip

{\bf Communication complexity for secret  key agreement.} In the protocol in Theorem~\ref{t:main}, Alice and Bob exchange $\min (C(x \mid y), C(y \mid x)) + O(\log n)$  bits and obtain with high probability  a shared secret key of length $I(x:y) - O(\log n)$. In this protocol we can assume that Alice and Bob use either private random bits or public random bits. We show that for the model with public random bits, the communication complexity of the protocol is optimal, in the sense that in any protocol with public random bits there are input strings $x$ and $y$, on which Alice and Bob have to exchange  at least $\min(C(x \mid y), C(y \mid x))$ bits. In fact, our lower bound is stronger: we show that, for any constants $\delta_1,\delta_2 > 0$, if Alice and Bob use a protocol with communication complexity $(1-\delta_1) \min(C(x \mid y), C(y \mid x))$ for every input pair $x,y$, then there are inputs for which the shared secret key that they obtain has length at most $\delta_2 I(x:y)$.
That is, if the communication complexity sinks below the threshold  $\min(C(x \mid y), C(y \mid x))$, then the size of the common secret key drops to virtually zero. Finding the optimal communication complexity for the model with private random bits remains an open problem.

\subsection{Related previous work.} 

\textbf{IT $\mbox{vs.}$ AIT.} Before reviewing existing related results in the IT and the AIT frameworks, it is useful to understand the distinction between the two theories.
%
%
In computer science the attribute \emph{random}  is  mainly used  in two (fairly different) contexts: 
\emph{random processes} and \emph{random objects}. In short, IT, which we also call Shannon's framework, focuses on the former, whereas AIT, which we also call Kolmogorov's framework, focuses on the latter.
On the one hand, we may think of an uncertain physical process with unpredictable  outcomes, and
employ the framework of the classic probability theory (distributions, random variables, etc.).
The notion of a \emph{random variable} formalizes the idea of  a process like coin tossing.  In this context we can  measure  the uncertainty of a random variable as a whole (by its Shannon's entropy,   its min-entropy, etc.), but  we cannot ask whether one specific outcome  is random or not. 
On the other hand, people use \emph{tables of random numbers}, which are  available as  specific sequences of digits, written on a disc or printed on a paper. 
The usefulness  of such a table depends on its individual properties:  frequencies of digits, presence or absence of hidden regularities,  compressibility, etc. The usual way to measure  the uncertainty of an individual string of digits is Kolmogorov complexity.
In both  contexts, the formal measures of randomness may or may not  involve computational complexity  (see, e.g., different versions of pseudoentropy for distributions and the resource bounded variants of Kolmogorov complexity for individual strings).  These two formalizations   of  randomness are connected but  not interchangeable.

Both notions of randomness appear in cryptography.
For example, in the one-time pad scheme,  two parties share a `random' key that remains `secret' for the attacker.  It is common to use Shannon's framework, and therefore the notions of randomness and secrecy are defined in terms of {random processes}. In the ideal situation both parties should have access to a common source of randomness, e.g., to the results of  tossing an unbiased coin (hidden from the adversary). By tossing this coin $n$ times we get a random variable with maximal possible entropy, and thus, in Shannon's framework,  the quality of randomness is perfect. But if by chance we obtain a sequence of $n$ zeros,  then this specific one-time pad looks pretty useless in any practical application. However, Shannon's information theory provides no vocabulary to complain about this apparently non-random \emph{individual} key. Antunes \emph{et al.} \cite{antunes-laplante} suggested to use Kolmogorov complexity to measure the `secrecy' of individual instances of a one-time pad or a secret sharing schemes. We have in mind a similar motivation, and in this work a ``secret key" is an  individual string that is random in the sense of Kolmogorov complexity. 
\smallskip

\textbf{Related work.} 
We start with a brief account of  works on secret key agreement in the IT setting. 
The secret key agreement is a relatively well-studied problem in information theory, motivated, as the name suggests, by applications in information-theoretically secure cryptography.  Wyner~\cite{wyn:j:tap} and Csisz{\'{a}}r and K\"{o}rner~\cite{csi-kor:j:secret} have analyzed the possibility of obtaining a shared secret key in the case when one party sends to the other party a single message on a channel from which the eavesdropper can obtain partial information. Maurer~\cite{mau:j:cryptoprov,mau:j:seckey} considered the case of protocols with several rounds of communication  and showed that interaction can be more powerful than one-way transmission. Ahlswede and  Csisz{\'{a}}r~\cite{ahl-csi:j:seckeyone} and Maurer~\cite{mau:j:seckey} have established the tight relation between interactive secret key agreement and mutual information for \emph{memoryless} sources.  In the memoryless model, the input data is given by two random variables $(X_1, X_2)$ obtained by $n$ independent draws from a joint distribution, where Alice observes $X_1$ and Bob observes $X_2$.  Informally stated, the references~\cite{mau:j:seckey, ahl-csi:j:seckeyone}  show that the longest shared secret key that Alice and Bob can establish via an interactive protocol with an arbitrary number of rounds is equal to the mutual information of $X_1$ and $X_2$.  Csisz{\'{a}}r and Narayan~\cite{csi-nar:j:seckey} go beyond the scenario with  two parties,  and consider the case of an $\ell$-memoryless source $(X_1, \ldots, X_\ell)$ and $\ell$  parties, each one observing one component of the tuple. They show that the longest shared secret key the $\ell$ parties can establish via an interactive protocol with an arbitrary number of rounds is equal to the entropy  $H(X_1, \ldots, X_\ell)$ of the $\ell$-memoryless source from which one subtracts  the minimum communication for omniscience.  Their result holds also for stationary ergodic sources, which generalize memoryless sources. As one can see, our results are very similar. They have been inspired by the papers~\cite{ahl-csi:j:seckeyone,mau:j:seckey,csi-nar:j:seckey} and represent the AIT analogue of the results presented above. As we explain in Section~\ref{subsection-shannon-vs-kolmogorov}, our results imply their IT analogues, and can be viewed as  more general because they do not require the memoryless or ergodicity properties of sources (in fact they do not require any generative model at all).   The question of finding the communication complexity of secret key agreement protocols in the IT framework has been raised by Csisz{\'{a}}r and Narayan~\cite{csi-nar:j:seckey}. Tyagi~\cite{tyagi2013common}  has shown that for memoryless sources it is equal to the difference between interactive common information in Wyner's sense and  mutual information. In Section~\ref{subsection-shannon-vs-kolmogorov} we explain how Tyagi's result compares to our results on communication complexity in the AIT framework.     

Let us now say a few words about related results from the AIT world. To the best of our knowledge, in AIT there has been no previous work on  secret key agreement. However, the general idea of ``materialization'' of  mutual information was studied extensively. 
Motivated by the intuition that mutual information represents the amount of shared information in two strings, researchers have explored the extent to which mutual information can be materialized more or less effectively. The relevant concept is that of \emph{common information.} Informally, a string $z$ is a common information string extracted from strings $x$ and $y$, if $z$ can be ``computed" from $x$, and also from $y$, where ``computed" is taken in a more liberal sense that allows the utilization of a few help bits. In the most common setting of parameters, we require that $C(z \mid x) = O(\log n)$ and $C(z \mid y) = O(\log n)$, where $n$ is the length of $x$ and $y$ and the constant hidden in the $O(\cdot)$ notation depends only on the universal machine. 
 Informally, the common information of $x$ and $y$ is the length of a longest common information string that can be extracted from $x$ and $y$.
It can be shown that up to logarithmic precision common information is upper bounded by mutual information. 
In an influential paper, G\'{a}cs and K\"{o}rner~\cite{gac-kor:j:commoninfo} have constructed strings $x$ and $y$ for which the common information is much smaller than the mutual information. Moreover, the property of a pair $(x,y)$ of having common information equal to mutual information does not depend solely on the complexity profile of $x$ and $y$: There exist pairs $(x_1, y_1)$ and $(x_2, y_2)$ having the same complexity profile, and for $(x_1,y_1)$ the common information and mutual information are equal, whereas for $(x_2, y_2)$ they are not. Muchnik~\cite{muc:j:commoninfo} and Romashchenko~\cite{rom:j:mutualinfo} have strengthened the G\'{a}cs-K\"{o}rner theorem in significant ways, by allowing a larger amount of help bits, parameterizing the mutual information of the constructed pair $(x,y)$, and other ways. Chernov et al.~\cite{che-muc-rom-she-ver:j:commoninfo} presents alternative constructions of strings for which the common information is smaller than mutual information for several regimes of parameters. A nice, self-contained and accessible exposition of this research line can be found in the book of Shen, Vereshchagin and Uspensky~\cite[Chapter 11]{suv:b:kolmenglish}.
 
 Thus, previous work has shown negative results regarding the ``materialization" of mutual information in AIT. As far as we know, ours is the first positive result. In summary, we now know that  computation without communication, even enhanced with help bits, fails to extract the mutual information of two strings, while interactive computation succeeds.
\medskip

\textbf{Our techniques.} It is common for statements in IT (in the Shannon's entropy framework) to have an analogous version in AIT (in the Kolmogorov complexity framework). However, there is no canonical way to translate a result from one setting to the other, and proofs of homologous results in these two frameworks can be drastically different. A textbook example of this phenomenon is the chain rule: it is valid for Shannon's entropy and for Kolmogorov complexity, and the  formal expressions of this rule in both frameworks look very similar. However, in Shannon's case this fact is an easy corollary of the definition, while in Kolmogorov's version it requires a nontrivial argument (which is known as the Kolmogorov--Levin theorem). There are more advanced examples of parallel properties (from IT and AIT respectively), where the discrepancy between their proofs is even more striking.

This phenomenon manifests itself in this work as well.  Our main results are motivated  by similar ones in IT, and there is a close resemblance of statements. As discussed above,  this is not surprising. In what follows we explain the relation between our proofs and the proofs of similar statements in Shannon's framework.

The  positive results (the existence of communication protocols) use constructions that at the high level are  akin to those from their IT counterparts~\cite{mau:j:seckey,ahl-csi:j:seckeyone,csi-nar:j:seckey}.  We employ a similar intuitive idea --- manipulations with  ``fingerprints'' of appropriate lengths,  see  Section~\ref{s:warmup}. However, the technical machinery is different. In the AIT framework, for communication-efficient protocols, we need  quite explicit constructions, 
while homologous results in IT are usually proven  by choosing random encodings. 
Our constructions are based on a  combination of extractors and universal hashing.
Specifically,  we use results from ~\cite{bau-zim:c:linlist} and~\cite{zim:c:kolmslepianwolf}, where the  ``digital fingerprints''   are obtained by composing certain  randomness extractors and universal hashing.

Our general protocols are not time-efficient, and this is to be expected given the high generality of the type of data correlation in the AIT setting.  However, for some particular types of correlation  (e.g., for a pair of inputs with a bounded Hamming distance), our protocols can be modified to run in polynomial-time. In this case, we use the reconciliation technique  from \cite{smith2007scrambling,guruswami2010codes,guruswami2016optimal}.

In the negative results (upper bounds for the size of the common secret key, Theorems~\ref{t:lower} and~\ref{t:upperbdthree}), the ideas from IT do not help.  The reason is that in the AIT framework the mutual information of various strings is not {exactly} zero but only close to zero within some slack terms.  The slack terms are small, but during the rounds of a protocol, the errors can accumulate and grow beyond control (for more details see the discussion of the limits of the  ``weak'' upper bound in Section~\ref{s:warmup}).
To overcome this obstacle,  we come up with a new type of inequalities for Kolmogorov complexity. These inequalities  are substantially different from the classic information inequalities used in the analogous results in IT. 
This technique (Lemmas~\ref{lemma-I(a:b:t)-is-positive-for-transcripts} and \ref{lemma-j(a,b,c|t)}) is  based on  ideas similar to the \emph{conditional information inequalities} in  \cite{kaced2013conditional,kaced2015conditional}. 
We believe that this technique can be helpful in other cases, including applications in IT 
(see Section~\ref{s:final}).

In the proof of a lower bound for communication complexity (Theorem~\ref{thm:comm-complexity}), we use methods  specific for AIT,  with no apparent parallel in IT. We adapt the technique of bounds for the size of common information that goes back to An.~Muchnik and use deep results regarding 
the contrast between  \emph{mutual information} and ``extractable'' \emph{common information} for 
stochastic for pairs of strings~\cite{shen1983concept,muchnik-romash,razenshteyn-2011}, which have not been previously employed in information theory and communication complexity.

\medskip

\textbf{Roadmap.} In Section~\ref{s:prelims} we present the necessary basic concepts and results in Kolmogorov complexity theory used in this paper, and we formally define secret key agreement protocols. Section~\ref{s:warmup} is a warm-up section and contains a somewhat simplified version of the main result with a relatively short and self-contained proof. In Section~\ref{s:mutualinfo} we prove our main result. Section~\ref{section:3-sources} contains the results on $\ell$-parties secret key agreement protocols, for $\ell \geq 3.$  Section~\ref{s:comm} presents the results on the communication complexity of secret key agreement protocols. Section~\ref{s:technical} contains the proofs of some technical lemmas. We conclude with some loose-ends and final comments in Section~\ref{s:final}. 

\medskip

\section{Preliminaries}
\label{s:prelims}
\subsection{The basics of algorithmic information theory}
\label{s:ait}
Algorithmic Information Theory (AIT), initiated independently by Solomonoff~\cite{sol:j:inductive}, Kolmogorov~\cite{kol:j:kolmcomplexity}, and Chaitin~\cite{cha:j:length-of-programs}, is a counterpart to the Information Theory (IT), initiated by Shannon.  While in IT, a string is the realization of a random variable,
AIT dispenses with the stochastic generative process, and defines the complexity of an individual string $x$ as the length of its shortest description. For example, the string 
\[x_1 = 00000000 00000000 00000000 00000000\]
 has low complexity because it can be succinctly described as ``$2^5$ zeros." The string 
 \[x_2 = 10110000 01010111 01010100 11011100\] 
 is a $32$-bit string obtained using random atmospheric noise (according to random.org), and has high complexity because it does not have a short description.

Formally, given a Turing machine $M$, a string $p$ is said to be a \emph{program} (or  a \emph{description}) of a string $x$, if $M$ on input $p$ prints $x$. We denote the length of a binary string $x$ by $|x|$. The \emph{Kolmogorov complexity} of $x$ relative to the Turing machine $M$ is
\[
C_M(x) = \min \{|p| \mid \mbox{ $p$ is a program for $x$ relative to $M$}\}.
\]
If $U$ is universal Turing machine, then for every other Turing machine $M$ there exists a string $m$ such that $U(m, p) = M(p)$ for all $p$, and therefore for every string $x$,
\[
C_U(x) \leq C_M(x) + |m|.
\]
Thus, if we ignore the additive constant $|m|$, the Kolmogorov complexity of $x$ relative to $U$ is minimal. We fix a universal Turing machine $U$, drop the subscript $U$ in $C_U(\cdot)$, and denote the complexity of $x$ by $C(x)$. We list below a few basic facts about Kolmogorov complexity and introduce some notation:
\begin{enumerate}
\item For every string $x$, $C(x) \leq |x| + O(1)$, because  a string $x$ is trivially described by itself. (Formally,  there is a Turing machine $M$ that, for every $x$,  on input $x$ prints $x$.)
\item Similarly to the complexity of $x$, we define the complexity of $x$ conditioned by $y$ as
$C(x \mid y) = \min\{|p| \mid \mbox{ $U$ on input $p$ and $y$ prints $x$}\}.$
\item Using some standard computable pairing function $\langle \cdot, \cdot \rangle$ that maps pairs of strings into single strings, we define  the complexity of a pair of strings   by $C(x,y) = C(\langle x,y\rangle)$. Then we can extend this notation to tuples of larger arity.
\item We use the convenient shorthand notation $ a \leqp b$ to mean that $a \leq b + O(\log n)$, where $n$ is a parameter that is clear from the context and the constant hidden in the $O(\cdot)$ notation depends on the universal machine $U$ and sometimes on some  computable functions specified in the text, but not on $a$ and $b$.  Similarly,  $a \geqp b$ means $a \geq b - O(\log n)$, and $a \eqp b$ means ($a \leqp b$ and $a \geqp b$). In several places we use $\leqp, \eqp, \geqp$ to indicate a loss of precision of $O(\log (n/\epsilon))$ (instead of $O(\log n)$), where $\epsilon$ is an additional parameter that is clear from the context.
\item The chain rule in information theory states that $H(X,Y) = H(X) + H(Y \mid X)$. A similar rule holds true in algorithmic information theory: for all sufficiently long strings $x$ and $y$,
\begin{equation}
\label{e:eqone}
\big| C(x,y)  - (C(x) + C(y \mid x)) \big|  \leq 3 (\log C(x) + \log C(y)).
\end{equation}
\item The mutual information of two strings $x$ and $y$ is denoted $I(x:y)$, and is defined as $I(x:y) = C(x) + C(y) - C(x,y)$. The conditional mutual information is defined in the analogous way: $I(x:y \mid z) = C(x \mid z) + C(y \mid z) - C(x,y \mid z)$. It can be shown that mutual information is essentially non-negative (up to an additive term $O(\log C(x,y))$) and it satisfies the chain rule $I(x,z : y) \eqp I(z:y) + I(x:y \mid z)$ (here the $n$ hidden in the $\eqp$ is $\max(|x|,|y|,|z|)$).
\item  The complexity profile of a tuple of strings $(x_1, \ldots, x_\ell)$ is given by the tuple consisting of the complexities
of all non-empty subsets of the strings in the tuple, \ie, it is the tuple $(C(x_V ) \mid V \subseteq [\ell], V \not= \emptyset)$.
We use here and elsewhere in the paper the notation
 $[\ell]$ for the set $\{1, \ldots, \ell\}$ and for an $\ell$-tuple $(x_1, \ldots, x_\ell)$ and $V$ a subset of $[\ell]$, $x_V$ denotes the subtuple obtained by taking the components with indices in $V$ (for example if $V =\{1,2,7\}$ then $x_V = (x_1, x_2, x_7)$).

\item 
The \emph{extended} complexity profile of a tuple of strings is the tuple of all conditional and unconditional complexities involving these strings. E.g., the extended complexity profile of a pair $(x,y)$ consists of the values
 $
  C(x), C(y), C(x,y), C(x\mid y), C(y\mid x).
 $

\item  We use the following stronger variant, due to Bruno Bauwens~\cite{bau:t:kolmslepwolf}, 
of the main  result from reference~\cite{zim:c:kolmslepianwolf}. It is a version of the Slepian-Wolf theorem about distributed compression of $\ell$ sources in the setting of Kolmogorov complexity. 
\begin{theorem}[Distributed compression of $\ell$ sources]
\label{t:compressmulti}
For every constant positive integer $\ell$, there exist a probabilistic encoding algorithms $E$ and a decoding algorithm $D$ such that for every $n$ for  every tuple of $n$-bit strings $(x_1, \ldots, x_\ell)$ and for every tuple of integers $(n_1, \ldots, n_\ell)$, 
\begin{itemize}
\item For every $i \in [\ell]$, $E$ on input $x_i, n_i$  and $\epsilon$ uses $O(\log (n/\epsilon))$ random bits and outputs a string $p_i$ of length $n_i + O(\log (n/\epsilon))$,
\item For every string $y$, if the numbers $n_i$ satisfy the conditions $\sum_{i \in V} n_i \geq C(x_V \mid y, x_{[\ell] - V})$ for every $V \subseteq [\ell]$, then the decoding algorithm $D$ on input $y$ and $p_1, \ldots, p_\ell$ reconstructs $x_1, \ldots, x_\ell$ with probability $1- \epsilon$, where the probability is over the randomness used by the encoding procedures.
\end{itemize}
\end{theorem}
\if01
Note:  In the paper~\cite{zim:c:kolmslepianwolf} the theorem is stated with $y$ being the empty string, but the proof works for any $y$. Also the 
paper~\cite{zim:c:kolmslepianwolf} contains actually a  result with different parameters:  the algorithm runs in polynomial time, but the price is that the $O(\log n)$  terms are replaced by polylog($n$). The above parameters can be obtained  with the same proof as in~\cite{zim:c:kolmslepianwolf}, except that the explicit extractor used there is replaced with an extractor built with the standard probabilistic method.
\fi
The following is the particular case of the above theorem for $\ell = 1$. A weaker version  appeared in reference~\cite{bau-zim:c:linlist}.

\begin{theorem}[Single source compression]
\label{t:compression}
There exists a probabilistic algorithm that on input $x,k, \epsilon$, where $x$ is a string, $k$ is a positive integer and $\epsilon > 0$, returns a string $p$ of length $k+ O(\log (n/\epsilon)) $, and for every $y$,  if $k \geq C(x \mid y)$, then with probability $(1-\epsilon)$, $p$ is a program for $x$ given $y$. The algorithm is using $O(\log (n/\epsilon))$ random bits, where $n$ is the length of $x$.
\end{theorem}

\end{enumerate}
\subsection{Shared secret keys and protocols for secret key agreement}
\label{s:defprotocol}

 Let $k$ be a positive integer. A $k$-rounds  two-party protocol for secret key agreement uses two computable functions $A$ and $B$ and runs as follows. The first party has as input a string $\xa$ and uses private randomness $\ra$,  the second  party has as input a string $\xb$ and uses private randomness $\rb$.   
We assume that the length of $\ra$ ($\rb$) is determined by $\xa$ (respectively, $\xb$).
The protocol consists of the following calculations:
\[
\begin{array}{rl} 
x_1 & = A(\xa, \ra) \\
y_1 & = B(\xb, \rb, x_1) \\
x_2 & = A(\xa, \ra, y_1) \\
y_2 & = B(\xb, \rb, x_1, x_2) \\

\vdots \\

x_k & = A(\xa, \ra, y_1, \ldots, y_{k-1}) \\
y_k & = B(\xb, \rb, x_1, \ldots, x_k).
\end{array}
\]
The algorithms $A$ and $B$ can handle inputs of different  lengths. We also allow them to be partial (\ie, it is possible that the protocol does not converge for some pairs of inputs).
Let  $\epsilon$ be a positive constant and $\delta(n)$ be a constant or a slow growing function ($\mbox{e.g.}$, $O(\log n)$). A protocol \emph{succeeds} with error probability $\epsilon$ and randomness deficiency $\delta(n)$ on a pair $(\xa, \xb)$ of $n$-bit strings  if with probability $(1-\epsilon)$ over $\ra, \rb$, 
\begin{equation}
\label{e:e1}
A(\xa,\ra, t) = B(\xb, \rb, t) \stackrel{def.}{=} z,
\end{equation}

and
\begin{equation}
\label{e:e2}
C(z \mid t) \geq  |z| - \delta(n),
\end{equation}
where $t=(x_1, y_1, \ldots, x_k, y_k)$ is the transcript of the protocol.

The string $z$ satisfying  equation~\eqref{e:e1} and inequality~\eqref{e:e2} is called a \emph{shared secret key}  output by the protocol on input $(\xa, \xb)$.
Note that the shared secret key $z$ is  a random variable since it depends not only on the inputs $x_A$ and $x_B$, but also on the randomness $\ra$ and $\rb$.

The number of rounds in a protocol (parameter $k$) may depend on the length of the inputs.

In words, Alice and Bob start with input strings $\xa$ and respectively $\xb$, use private randomness $\ra$, and respectively $\rb$ and execute a protocol in which at round $i$, first Alice sends to Bob the string $x_i$,  and next Bob sends to Alice the string $y_i$, and  at the end Alice and Bob separately compute with high probability a common string $z$ (equation~\eqref{e:e1}) such that $z$ is random even conditioned by the transcript  of the  protocol  (inequality~\eqref{e:e2}). Thus, $z$ is secret to an adversary that has observed the protocol and consequently knows the transcript. 

\begin{remark}\label{remark-on-valid-inputs}
In our secret key agreement protocols, the inputs $\xa$ and $\xb$ have two components: $\xa = (x, h_A)$ and $\xb = (y, h_B)$, where the strings $x$ and $y$ are the main components, while $h_A$ and $h_B$ are short helping strings (for example, containing information about how $x$ and $y$ are correlated). The communication protocols designed in this paper succeed for all input pairs $\xa$ and $\xb$  in which $h_A = h_B = \mbox{\textup(the complexity profile of $x$ and $y$\textup)}$. In case one or both of $h_A$ and $h_B$ are not equal to the complexity profile, the protocols still halt on every input, but the outputs may be meaningless. 
\end{remark}

\begin{remark}\label{remark-on-non-uniform-protocols}
In this paper we deal with \emph{uniform} communication protocols: for inputs of all lengths Alice and Bob compute their messages and the final results by using the same algorithms $A$ and $B$. Our main results can be easily extended to non-uniform protocols,  with ``custom-made'' algorithms $A$ and $B$ for each length of inputs. In this setting the technical statements of the theorems would be more cumbersome (we would need to add the description of the protocol in the condition of all complexity terms, assuming that the protocol is known to the adversary; 
for details, see Remarks~\ref{remark-on-non-uniform-protocols-bis},~\ref{remark-on-non-uniform-protocols-bis-bis}), 
though the adjustment of the proofs  is quite straightforward.
\end{remark}

\section{Warm-up:  Short proofs for ``light"-versions of the main results}
\label{s:warmup}
Our main result has two parts: the \emph{lower bound}, which consists  of a protocol that produces a shared secret key of length equal (within logarithmic precision) to the
 mutual information of the inputs, and the \emph{upper bound}, where we show that it is impossible to obtain a longer shared secret key. We present here \emph{light}  versions of these two parts. We call them \emph{light} because they have weaker parameters than the full-fledged results, which are presented in Section~\ref{s:mutualinfo}. On the other hand, they admit short and self-contained proofs.  Moreover, since they have a common structure, these proofs are a useful preamble for the proofs in the next section.
 
 The \emph{light lower bound} consists of a protocol that does allow two parties, one having the $n$-bit string $x$ and the other having the $n$-bit string $y$, to obtain a shared secret key of length approximately equal to $I(x:y)$, but which communication-wise is inefficient. In Section~\ref{s:mutualinfo} we present an improved protocol that achieves the same task, with communication complexity equal essentially to $\min (C(x \mid y), C(y \mid x))$, which we conjecture to be optimal (see Section~\ref{s:comm}).
 The \emph{light upper bound} consists of a proof in which we show that no protocol with a constant number of rounds and $\poly(n)$ random bits  can produce a secret key longer than the mutual information. In Section~\ref{s:mutualinfo}, we lift the restrictions regarding the number of rounds and random bits.
 \medskip

  \textbf{The light lower bound}.\footnote{The following proof is due to Bruno Bauwens (private communication, August 2017).}
 The setting is that Alice has the $n$-bit string $x$, and Bob has the $n$-bit string $y$.  In addition, both Alice and Bob have the complexity profile of $x$ and $y$.  We construct a protocol that allows Alice and Bob to agree with probability $1-\epsilon$ on a shared secret key of length $I(x:y) - O(\log(n/\epsilon)$ and randomness deficiency $\delta(n) = O(\log(n/\epsilon))$. The protocol has $1$ round, and actually the communication consists of a single message from Alice to Bob. 

\label{communication-protocol--nutshell}
We start  with an overview of the proof in a ``nutshell.''  First, Alice chooses  two \emph{hash-functions} $h_1$ and $h_2$ from two families of hash functions which are publicly known. Alice sends $h_1$ and $h_2$ to Bob, and she  also sends to Bob the hash value $h_1(x)$. Next, Bob attempts to reconstruct Alice's input: he searches over all $x'$ such that $C(x'\mid y) \le C(x\mid y)$ until he finds a candidate with the same hash value, $h_1(x')=h_1(x)$; in what follows he hopes that the $x'$ he found  is indeed Alice's string $x$.   Next, Alice and Bob compute the value $h_2(x)$ and take it as the common secret key.
 
 In this scheme the adversary gets to know the hash functions $h_1, h_2$ and the hash-value $h_1(x)$.  For this to work, we want to guarantee that with high probability Bob finds the correct value of $x$ (absence of collisions) and that the adversary gets very little information regarding the secret $h_2(x)$.  The technical part of the proof consists in choosing appropriate families of hash functions, which guarantee the two required properties.  
 
 In this section we use a very simple hashing technique --- we take $h_1$ and $h_2$ to be random linear mappings with  suitable lengths of the output. In the next section we use  Theorem~\ref{t:compression},  which provides a subtler implementation of the same idea,  where the `hashing' is based on extractors.
 
 \smallskip
 
 In what follows we present in more detail the easy version of the protocol, with linear hashing.



 Step $1$: Let $n_1 = C(x)$ and let  $k = C(x \mid y)$.  (In fact, Alice  does not have $k$. But she can compute from the complexity profile $x$ and $y$, a good approximation of $k$; this is a small nuisance which can be easily handled, and for the sake of simplicity we ignore it.)  Alice chooses an $n_1+\log(1/\epsilon)$-by-$n$ random binary matrix $H$. Alice partitions $H$ into two sub-matrices $H_1$ and $H_2$, where $H_1$ consists of the top $k+\log(1/\epsilon)$ rows of $H$, and $H_2$ consists of the bottom $n_1-k$ rows of $H$.  Alice calculates $q= H_1 x$, and sends $q$ and the entire $H$ to Bob.  Let $E_1$ be the event that there exists $x' \not= x$ such that $C(x') \leq C(x)$ and $Hx'=Hx$ and $E_2$ be the event that there exists $x''\not= x$ such that $C(x'' \mid y) \leq C(x\mid y)$ and $H_1 x'' = H_1x$. Note that both $E_1$ and $E_2$ have probability at most $\epsilon$.  Consequently with probability at least $1-2\epsilon$,  neither $E_1$ nor $E_2$ hold,  an event that we  assume henceforth.  We also assume 
 \begin{equation}
 \label{e:randH}
 C(x,y \mid H ) \geq C(x,y)  - O(\log(1/\epsilon)),
 \end{equation} an event which holds with probability at least $1-\epsilon$.
  \smallskip
 
 Step $2$:  Since $x$ is the unique string such  $H_1 x = q$ among strings that have complexity conditioned by $y$ bounded by $k$, Bob reconstructs $x$ from $q$ and $H$. Next, both Alice and Bob compute the string $z= H_2 x$, which is the output of the protocol.  
 
 Let us now show that $z$ has the desired properties. Note, first,  that $z$ has length $n_1 - k = C(x) - C(x\mid y) \eqp I(x:y)$. It remains to show that $z$ is a shared secret key, \ie, that $z$ is random conditioned by the transcript of the protocol, which consists of $q$ and $H$. More precisely, we  show that $C(z \mid q, H) \geq I(x:y) - O(\log(n/\epsilon)$.
  Since $x$ can be computed from $q,z$ and $H$,  
 \begin{equation}
 \label{e:lightlb1}
 C(x,y \mid H) \leq C(q,z,y \mid H) +O(1).
 \end{equation}
 Next, by the chain rule,
 \begin{equation}
 \label{e:lightlb2}
 \begin{array}{rl}
 C(q,z,y \mid H) & = C(q \mid H) + C(z \mid q, H) + C(y \mid q,z,H) + O(\log n) \\
 & \leq C(x \mid y) + C(z \mid q, H) + C(y \mid x) + O(\log (n/\epsilon)).
 \end{array}
 \end{equation}
 In the second line, we have taken into account that $C(q \mid H)$ is less than $|q| + O(1) = C(x \mid y) + \log(1/\epsilon)+ O(1)$, and $C(y \mid q,z, H)$ is less than $C(y \mid x) + O(1)$ (because $x$ can be computed from $q,z, H$).
 Combining equations~\eqref{e:lightlb1} and~\eqref{e:lightlb2}, we obtain
 \[
 \begin{array}{rl}
 C(z \mid q, H) & \geq C(x, y \mid H) - C(x \mid y) - C(y \mid x) - O(\log(n/\epsilon)) \\
 & \geq C(x, y) - C(x \mid y) - C(y \mid x) - O(\log(n/\epsilon))  \mbox{ \quad (using equation~\eqref{e:randH}) }\\
 & \geq I(x: y)  - O(\log(n/\epsilon)).
 \end{array}
 \]
\bigskip

\begin{remark}The defined protocol has communication complexity $O(n^2)$ (mostly due to the necessity to communicate  the matrix $H$). The communication complexity of the protocol can be easily reduced to $O(n)$ if we take at random, not an arbitrary matrix $H$, but a random Toeplitz matrix of the same size. In the next section we describe a protocol with communication complexity $\min (C(x \mid y), C(y \mid x))$.  
\textup(We show in Section~\ref{s:comm} that no protocol with public randomness has smaller communication complexity, and we conjecture that the same holds for protocols with private randomness.\textup) 
\end{remark}

\textbf{The light upper bound.} 
We show here that if $k$ is a constant and the number of random bits is bounded by a polynomial in $n$, no $k$-round protocol can produce a shared secret key that is longer (up to logarithmic precision) than the mutual information of the inputs. More precisely, our claim is that if a shared secret key $z$ is produced from inputs $\xa, \xb$ by some $k$-round protocol with probability $1-\epsilon$, then $C(z \mid t) \leq  I(\xa:\xb) + O(k\log(n/\epsilon))$ with probability $1 - O(\epsilon)$, where $t$ is the transcript of the protocol. 

\emph{Simple warm-up: deterministic protocols without communication.}
 Let us first consider the case $k=0$, \ie, Alice and Bob do not interact, and also, for now, let us  assume they use no randomness. 
 Clearly, this model is very weak, and, in general, Alice and Bob cannot agree on any nontrivial common secret without communication. However, in some cases (for some very special inputs $\xa$ and $\xb$) this might be possible. For instance, if $\xa$ and $\xb$ are $n$-bits strings of complexity $n$, and the first   $n/2$ bits in these strings coincide (the $i$-th bit of $\xa$ is equal to the $i$-th bit of $\xb$ for $i=1,\ldots,n/2$), then Alice and Bob can take the first $n/2$ bits of their inputs as the common key. 
 
 In this trivial example,  Alice and Bob  get a common key of complexity close to $n/2$, and  the mutual information between $\xa$ and $\xb$ is also not less than $n/2$. 
 We  argue that the same relation remains true for all protocols without communication, for all pairs of inputs $(\xa,\xb)$.
That is, if Alice and Bob can extract from their inputs $\xa$ and $\xb$ the same string $z$ without communication, then  the complexity of this $z$ is bounded by $I(\xa:\xb)$. This follows  from a well-known  inequality that is true for all $\xa,\xb,z$ (see Lemma~\ref{l:comm-inf}):
\begin{equation}\label{comm-inf}
C(z) \le^+ C(z \mid \xa) + C(z \mid \xb) + I(\xa:\xb).
\end{equation}
\if01
(For the sake of self-containedness we  sketch  the proof of \eqref{comm-inf} in  Lemma~\ref{l:comm-inf} in Section~\ref{s:technical}.)
\fi

 \emph{Slightly more freedom: randomized protocols without communication.}
The same bound remains true if Alice and Bob use random strings $\ra$ and $\rb$ (still without  communication).
Indeed, using Lemma~\ref{l:chainrandom}~(3), with probability $1-O(\epsilon)$, 
 we have  
\begin{equation}
\label{e:randeffect}
I(\xa,\ra : \xb,\rb)= I(\xa:\xb) \pm O(\log(n/\epsilon)).
\end{equation}
 Hence,  we can apply the  same argument to  the pair $(\langle \xa,\ra\rangle, \langle \xb,\rb \rangle)$.
\begin{equation}\label{comm-inf-with-ra-rb}
\begin{array}{rl}
C(z)  & \le^+ C(z \mid \xa,\ra) + C(z \mid \xb,\rb) + I(\xa,\ra:\xb,\rb) \\
& \eqp   C(z \mid \xa,\ra) + C(z \mid \xb,\rb) +  I(\xa:\xb) \pm O(\log(n/\epsilon))
\end{array}
\end{equation}
(the first inequality in \eqref{comm-inf-with-ra-rb} is always true, while the last equality is true  with probability $1-O(\epsilon)$.)
In case when $z$ is a computable function of $( \xa,\ra)$ and a  computable function of $( \xb,\rb)$, we obtain from \eqref{comm-inf-with-ra-rb} that with probability $1-O(\epsilon)$
\begin{equation}\label{comm-inf-with-ra-rb-simplified}
C(z) \leq  I(\xa:\xb) + O(\log(n/\epsilon)),
\end{equation}
and we are done. 

 \emph{The general case: randomized protocols with $O(1)$-round communication.}
Let us move now to the case when the number of communication rounds is $k > 0$.  As announced, we claim that  the upper bound  remains true: $C(z \mid t) \leq I(\xa: \xb) + O(\log(n/\epsilon))$. 
 The proof  for the ``light"  case in which $k$ is  constant proceeds as follows.

Inequality~\eqref{comm-inf} has a relativized version, 
\begin{equation}\label{comm-inf-relativ}
C(z \mid t) \le^+ C(z \mid \xa,\ra, t) + C(z \mid \xb,\rb, t) + I(\xa, \ra : \xb, \rb \mid t),
\end{equation}
which is true for all $\xa,\xb, \ra, \rb, z,t$. 
When $z$ is a computable function of $( \xa,\ra,t)$ and a  computable function of $( \xb,\rb,t)$, we obtain
\begin{equation}\label{comm-inf-relativ-simplified}
C(z \mid t) \le^+  I(\xa, \ra: \xb, \rb\mid t).
\end{equation}
Thus, to prove our claim,  it remains to show that for the transcript $t$ 
\begin{equation}\label{eq-0}
   I(\xa, \ra : \xb, \rb \mid t)  \le^{+}   I(\xa, \ra : \xb, \rb ).
\end{equation}
and then to use inequality~\eqref{e:randeffect}.

 Notice that, in general, the conditional mutual information $I(a : b \mid c)$ can be much larger than the plain mutual information $I(a:b)$ (e.g., if $a$ and $b$ are independent  strings of length $n$ and $c$ is the bitwise XOR of  $a$ and $c$, then $I(a:b\mid c) \eqp n  \gg I(a:b)\eqp 0$). So, we should explain why \eqref{eq-0} is true for these specific strings. The required inequality can be deduced from the following lemma.
 \begin{lemma}\label{lemma-I(a:b:t)-is-positive}
 (a) Let $f$ be a computable function of one argument. Then for all $a,b$
  \begin{equation}\label{I(a:b:t)-is-positive}
   I(a:b \mid f(a))  \le^{+}   I(a:b ).
  \end{equation}
(b) Further, let $g$ be a computable function of two arguments. Then for all $a,b,c$
  \begin{equation}\label{I(a:b:t|c)-is-positive}
   I(a:b \mid  g(a,c), c )  \le^{+}   I(a:b \mid c).
  \end{equation}
  \end{lemma}
\begin{proof} Inequality~\eqref{I(a:b:t)-is-positive} follows directly from the chain rule:
$$
\begin{array}{rcll}
I(a:b) &\eqp &  I(a,f(a):b)&\mbox{(since $f(a)$ can be computed from $a$)}\\
    &\eqp &  I(f(a):b) +  I(a:b \mid f(a))&\mbox{(the chain rule)}\\
    &\ge ^+& I(a:b \mid f(a)) &\mbox{(non-negativity of the mutual information)}
\end{array}
$$
The proof of \eqref{I(a:b:t|c)-is-positive} is similar, we should only substitute $c$ as a condition to each term $I(*)$ in the inference above. The lemma is proven.
\end{proof} 
 
\if01
 \smallskip
 
 \emph{Remark (of minor importance).} Inequality~\eqref{eq-0} rewrites equivalently to 
 $$I(x_A, r_A  : y, r_B : t ) \ge^+0,$$ 
 where $I(* : * : *)$ denotes the mutual information of a triple,  defined as 
 $$
 I(a:b:c) :=  C(a) + C(b) +C(c) - C(a,b) - C(a,c) - C(b,c) +C(a,b,c)
 $$
 In general, the mutual information of a triple can be far below zero. So our aim is to show that for this specific triple of objects the mutual information of the triple is non-negative.
  
\bigskip
 
  \paragraph{Paraphrase of the proof from the draft.} Inequality $ I(a:b:t) \ge^+0$ is true assuming  that $t$ is a deterministic function of $x$ or $y$. This fact is easy to prove (the chain rule). Thus, in case $t=\langle t_1,\ldots, t_k\rangle$, where each $t_i$ is a deterministic function of the previous $t_1,\ldots,t_{i-1}$ together with $x$ or $y$, we get
  $$
  I(x_A, r_A  : y, r_B | t_1,\ldots, t_k) \le^+   I(x_A, r_A  : y, r_B | t_1,\ldots, t_{k-1}) \le^+ \ldots \le^+   I(x_A, r_A  : y, r_B).
  $$
Technically, here we sum up $k$ inequalities, and each of them is valid up to an additive term $O(\log n)$. Thus, 
   $$
  I(x_A, r_A  : y, r_B | t_1,\ldots, t_k)  \le   I(x_A, r_A  : y, r_B) + O(k\log n).
  $$
 Since the number of rounds $k$ is constant, this proves the claimed upper bound. To prove the same inequality for a longer $k$ we use a different technique in Section~\ref{s:mutualinfo}. 
  \fi
  
Now we proceed with a proof of \eqref{eq-0}.
Recall from the definition of a protocol, that $t = (x_1,y_1, \ldots, x_k, y_k)$ and each $x_i$ (respectively $y_i$) is deterministically computed from $\xa, \ra$ (respectively from $\xb, \rb$) and the previous messages.
We claim that
\begin{equation}
\label{e:e3}
\begin{array}{rl}
I(\xa, r_A:\xb, r_B) 
                        & \ge ^+   I(\xa, r_A : \xb, r_B  \mid  x_1)  \\
                       & \geq^+  I(\xa, r_A : \xb, r_B \mid  x_1,y_1) \\
                        & \geq^+  I(\xa, r_A : \xb, r_B \mid  x_1,y_1, x_2) \\
			&\vdots  \\
			& \geq^+  I(\xa, r_A : \xb, r_B \mid  x_1,y_1, x_2,y_2, \ldots, x_k, y_k).
\end{array}
\end{equation}
 The first line is a special case of \eqref{I(a:b:t)-is-positive} since $x_1$ can be computed from  $(\xa,\ra)$. Every next line is a special case of \eqref{I(a:b:t|c)-is-positive}. Indeed, on each step we add to the condition of the term $I(\xa, r_A : \xb, r_B \mid  *)$  another string $x_i$ or $y_i$, and every time the newly added  string is a function of $(\xa, \ra)$ and the pre-history of the communication or  of $(\xb, \rb)$ and the pre-history.

All the inequalities used in the argument are valid up to an additive term $O(\log C(\xa,\ra,\xb,\rb))$, which is at most $O(\log n)$ because $\ra$ and $\rb$ have length bounded by a polynomial in $n$.   Since the number of rounds $k$ (and therefore the chain of inequalities \eqref{e:e3}) is constant, the claimed upper bound is proved.

To prove the same inequality for a non-constant $k$ and an arbitrary number of random bits, we will need a different technique. We discuss it  in Section~\ref{s:mutualinfo}.

 \section{Main  results}
\label{s:mutualinfo}
\if01
{\begin{definition}
\label{d:seccap}
The secrecy capacity of two strings $\xa,\xb$ relative to a secret key agreement protocol is the length of the secret key $z$ output by the protocol on input $\xa,\xb$. Note that since $z$ is a random variable, the secrecy capacity is also a random variable.
\end{definition}

First, we prove that there exists a $1$-round protocol having the features presented in the Introduction that allows two parties to agree on a secret key of length equal to $I(x:y)$  modulo $O(\log n)$ precision, with high probability.
 \fi
 We prove here our main results. We present a secret key agreement protocol which produces a shared secret key of length  equal (up to logarithmic precision) to the  mutual information of the inputs, provided the two parties know the complexity profile. Next, we show that no protocol can produce a longer shared secret key. The formal statements are as follows.

\begin{theorem}[Lower bound]
\label{t:lower}
There exists a secret key agreement protocol with
the following property: For every $n$-bit strings $x$ and $y$, for every constant $\epsilon > 0$,  if Alice's input $\xa$ consists of $x$,  the complexity profile
of $(x, y)$ and $\epsilon$, and Bob's input $\xb$  consists of $y$, the complexity profile of $(x, y)$ and $\epsilon$, then, with probability
$1 - \epsilon$,  the shared secret key is a string $z$ such that,
$C(z \mid t) \geq  |z| - O(\log(1/\epsilon))$ and $|z| \geq I(x : y) - O(\log( n/\epsilon))$,
where $t$ is the transcript of the protocol. 

Moreover, the communication consists of a single message sent by Alice
to Bob of length $C(x \mid y) + O(\log (n/\epsilon))$, Alice uses $O(\log(n/\epsilon))$ random  bits, and Bob does not use any random bits.
\end{theorem} 

\begin{theorem}[Upper bound] 
\label{t:upper}
Let us consider a protocol for secret key agreement, let $\xa$ and $\xb$ be input strings of length $n$  on which the  protocol succeeds with error probability $\epsilon$ and randomness deficiency $\delta(n)$, and let $z$ be the random string that
is the shared secret key output by the protocol, \ie, a string satisfying relations (\ref{e:e1}) and (\ref{e:e2}). Then with
probability at least $1-O(\epsilon)$, 
 if $n$ is sufficiently large, $|z| \leq I(\xa : \xb)+ \delta(n) + O(\log (n/\epsilon))$, where the constants in the $O(\cdot)$ notation depend
on  the universal machine and the protocol, but not on $\xa$ and $\xb$.
 \end{theorem}

Theorem~\ref{t:upper} establishes the upper bound claimed in the Introduction. Indeed, for any pair of $n$-bit strings $(x,y)$,  suppose that Alice's input $\xa$ consists of $x$ and the complexity profile of $(x,y)$ and Bob's input $\xb$ consists of $y$ and the complexity profile of $(x,y)$. Note that $I(\xa :  \xb) \eqp I(x:y)$ because the length of the complexity profile is bounded by $O(\log n)$. Hence, Theorem~\ref{t:upper} implies that secret key agreement protocols in which the two parties,  besides $x$ and respectively $y$, are additionally given the complexity profile of their inputs,  can not produce a secret key that is longer than $I(x:y) + O(\log n)$ (provided the randomness deficiency of the key satisfies $\delta(n) = O(\log n)$).

\begin{remark}
\label{remark-on-non-uniform-protocols-bis}
In  Theorem~\ref{t:lower} we construct  a uniform communication protocol (see Remark~\ref{remark-on-non-uniform-protocols} on p.~\pageref{remark-on-non-uniform-protocols}),  i.e., our protocol is described by  one program that deals with all $x$ and $y$ (the complexity profile of $x$ and $y$ is provided as an auxiliary input).  Theorem~\ref{t:upper} is also stated for uniform protocols. In the uniform case we can ignore whether the adversary knows the protocol (this knowledge affects all quantities of Kolmogorov complexity by at most an $O(1)$ term).  

Our negative result (Theorem~\ref{t:upper})  can be extended to non-uniform protocols. In the non-uniform setting the statement of Theorem~\ref{t:upper} has to be adjusted: we should assume that the protocol is ``known'' to the adversary, and all complexity terms are ``relativized''
with respect to the protocol.   Technically this means that  we  append the description of the protocol to the condition of all complexity quantities  involved in the theorem and  in the value of $C(z \mid t)$ in the definition of the randomness deficiency.
\end{remark}

\begin{proof}[Proof of Theorem~\ref{t:lower}]
We first give an overview of the protocol (an adaptation of the scheme ``in a nutshell'' explained on p.~\pageref{communication-protocol--nutshell}). 

 In Step 1, Alice computes $\pxy$ a program for $x$ given $y$ of length $C(x \mid y) + O(\log n)$.
She sends to Bob the string $\pxy$ and also  random  strings $r$ and $s$  that both Alice and Bob will use in Step 2.  The strings $r$ and $s$  have length  $O(\log (n/\epsilon))$.

In Step 2, Bob uses $\pxy$ and his string $y$, to construct $x$.  Next, both Alice and Bob construct a $O(\log n)$-short program $\tz$ for $x$ given $\pxy$ using the random string $r$. Then, using $s$, they modify $\tz$ and obtain the common output $z$ of the protocol.
\if01
First note that,
\[
\begin{array}{rl}
C(x \mid \pxy) & \eqp C(x, \pxy) - C(\pxy) \\
& \eqp  C(x) - |\pxy| \mbox{\qquad (because $C(x, \pxy) \eqp  C(x)$ and $C(\pxy)\eqp  |\pxy|$)} \\
&\eqp  C(x) - C(x \mid y) \eqp  I(x:y).
\end{array}
\]
Next, we have $C(z) \eqp  |z|$ (because $z$ is a short program, so it cannot be compressed by much), and $|z| \eqp  C(x \mid \pxy)$ (by the definition of $z$). Also, $C(z \mid \pxy) = ^+ C(x \mid \pxy)$, because given $\pxy$, one can construct $x$ from $z$, and vice-versa. Since the  transcript $t$ of the protocol is just $\pxy$, we obtain the desired conclusion.
\fi

Let us now see the details.  

Step 1.  Alice, using the complexity profile of $x$ and $y$, first computes $k = C(x,y) - C(x) + 4 \log n$. Taking into account relation~(\ref{e:eqone}), if $n$ is sufficiently large,
$C(x \mid y) +8 \log n \geq k \geq C(x \mid y)$. Next,  using the algorithm from Theorem~\ref{t:compression}, Alice on input $(x, k, \epsilon)$ computes a string $\pxy$ that, with probability $1-\epsilon$, is a program for $x$ given $y$ of length $k + O(\log (n/\epsilon)) \leq C(x \mid y) + O(\log (n/\epsilon))$. In what follows, we assume that $\pxy$ has this property (thus, we exclude an event that has probability at most $\epsilon$).

Note that $C(\pxy) \geq C( x \mid y)$ (because $\pxy$ is a program for $x$ given $y$), and  $C(\pxy) \leq |\pxy| + O(1) \leq k + O(\log (n/\epsilon)) \leq C(x \mid y) + O(\log (n/\epsilon))$. Thus, $C(\pxy) \eqp  C(x \mid y)$. 

Next, 
\begin{equation}
\label{e:eqq2}
\begin{array}{rl}
C(x \mid \pxy) &\eqp  C(x, \pxy) - C(\pxy) \mbox{\quad (chain rule)} \\

&= ^+ C(x) - C(x \mid y) \mbox{\quad ($\pxy$ is computed from $x$ and $O(\log(n/\epsilon))$ randomness)} \\

&\eqp  I(x:y).
\end{array}
\end{equation}

Alice sends to Bob the string $\pxy$ and also random strings $r$ and $s$  of length $O(\log(n/\epsilon))$ that will be used by both parties in Step 2.
\medskip

Step 2.  First,  Bob reconstructs  $x$ from $y$ and $\pxy$.

Next, both Alice and, separately,  Bob do the following. In the first phase, they use the random string $r$ (from Step 1) and compute a string $\tz$, which  almost has the desired properties, except that its randomness deficiency is a little too large. In a second phase, they use the random string $s$ (also from Step 1), to finally obtain $z$ with all the desired properties.

\emph{Phase 1.}  Using the complexity profile of $x$ and $y$ and based on relation~(\ref{e:eqq2}), Alice and Bob compute $k' = I(x:y) + c' \log (n/\epsilon)$, for a constant $c'$ such that $C(x | \pxy) \leq k'$.  Then, using the algorithm from Theorem~\ref{t:compression} on input $(x,k', \epsilon)$, with randomness given by the random string $r$ from Step 1, they compute a string $\tz$ of length $k' + O(\log (n/\epsilon))$. With probability (over $r$) at least $1-O(\epsilon)$, $\tz$ is a program for $x$ given $\pxy$. In what follows, we assume that $\tz$ has this property (thus, we exclude another event that has probability at most $O(\epsilon)$).

Note that
$|\tz| = k' + O(\log n) \eqp  I(x:y)$
and
$C(\tz) \leq |\tz| + O(1) \leq^+  I(x:y)$.
Also
\[
\begin{array}{rl}
C(\tz \mid \pxy) &\geq^+ C(x \mid \pxy) \mbox{\qquad\qquad (because $\tz$ is a program for $x$ given $\pxy$)} \\
&\eqp I(x:y). \mbox{\qquad \qquad  \quad (by using inequality~\eqref{e:eqq2})}
\end{array}
\]
Combining the last three relations, it follows that 
\[
|\tz|\eqp  C(\tz) \eqp  C(\tz \mid \pxy) \eqp  I(x:y).
\]
The transcript of the protocol (without $s$ which has not been used yet) is $t = (\pxy, r)$, and since $|r| = O(\log (n/\epsilon))$,  we get 
\begin{equation}
\label{e:tz}
C(\tz \mid t) \geq I(x:y) - O(\log (n/\epsilon)) \mbox{  and  } |\tz| =  I(x:y) - O(\log (n/\epsilon)).
\end{equation}
 Thus, we are almost done, except for the fact that the randomness deficiency of $\tz$ is $O(\log (n/\epsilon))$, which is slightly too large.
\smallskip

\emph{Phase 2.} To reduce the randomness deficiency, we need to use a strong extractor. 
In what follows we use the standard definitions of extractors  that can be found, e.g., in \cite{shaltiel2002recent}.
 Recall that a function  $E: \zo^n \times \zo^d \mapping \zo^m$ is a $(k, \epsilon)$ extractor if for any set $B \subseteq \zo^n$ of size $|B| \geq 2^k$, and for any $A \subseteq \zo^m$,
\[
\prob(E(U_B, U_{\zo^d}) \in A) \in \big(\frac{|A|}{M}- \epsilon, \frac{|A|}{M} + \epsilon \big),
\]
where $U_B, U_{\zo^d}$ denote independent random variables uniformly distributed  on $B$, and respectively $\zo^d$, and $M = 2^m$.

A function  $E: \zo^n \times \zo^d \mapping \zo^m$ is a $(k, \epsilon)$ \emph{strong} extractor if the function
$E_1$ defined by $E_1(x,y) = (y, E(x,y))$ is a $(k, \epsilon)$ extractor. Using the probabilistic method, one can show that there exist strong extractors with $m = k - 2 \log(1/\epsilon) - O(1)$ and $d = \log n  + 2 \log (1/\epsilon) + O(1)$~\cite[Theorem 6.17]{vad:b:pseudorand}. Alice and Bob compute  by brute-force such a strong extractor $E$. 

 For some value of $k = I(x:y) - O(\log(n/\epsilon))$ (the exact value of $k$ is specified in the proof of the next lemma), Alice and Bob use the strong $(k, \epsilon)$  extractor $E : \zo^{|\tz|} \times \zo^d \mapping \zo^m$ (with $d$ and $m$ as  above),  and the random string $s$ from Step  $1$, and take $z = E(\tz, s)$.  Intuitively, by the properties of  a strong extractor, the output $z$ is random even conditioned by the seed $s$, which is exactly what we need here. 
 The following lemma, whose proof  is given in Section~\ref{s:strongext}, 
 shows that indeed $z$ has the desired properties.
\begin{lemma}
\label{l:def-reduction}
Let $E$ be the above strong extractor, and suppose $\tz$ satisfies inequalities~\eqref{e:tz}. If $s$ is chosen uniformly at random in $\zo^d$, and $z = E(\tz, s)$, then  with probability $1-O(\epsilon)$, $C(z \mid t, s) \geq |z| - O(\log(1/\epsilon))$ and $|z| = I(x:y) - O(\log(n/\epsilon))$.
\end{lemma}
 This concludes the proof of the theorem.
\end{proof}

\begin{remark}
\label{r:lower-approx}
The proof of Theorem~\ref{t:lower} can be adapted to the situation where Alice and Bob are not given the exact value of the complexity profile of $(x, y)$ but only an approximation of this profile. 
Indeed, assume that Alice and Bob are given upper and lower bounds for each component of the complexity profile of $(x,y)$ with  precision  $\le \sigma$, for some integer $\sigma$.
Then we can use the communication protocol from the proof of Theorem~\ref{t:lower}  where the length of $\pxy$ is chosen in accordance with the \emph{upper bound} on $C(x \mid y)$, and the length of the resulting key $z$ is chosen in accordance with the \emph{lower bound} on the value of $I(x:y)$.
As a result,  with probability $1 - O(\epsilon)$,  Alice and Bob agree on a common secret $z$ that is incompressible \textup(i.e., $C(z \mid t) \geq  |z| - O(\log(1/\epsilon))$  where $t$ is the transcript of the protocol\textup),  and the length of $z$ is greater than $I(x:y) - \sigma - O(\log( n/\epsilon))$. 
\end{remark}

\begin{remark}
 In  the  protocol in the proof of Theorem~~\ref{t:lower}, we use  ``digital fingerprints'' of $x$ computed by a combination of randomness extractors and universal hashes  (a part of the construction of these fingerprints  is hidden in the proof of Theorem~\ref{t:compression}).
We apply ``digital fingerprints'' in two significantly different contexts:  in Step~1,   we use a fingerprint of $x$ to identify this string
given the partial information on $x$ held  by Bob, 
and in Step~2,   we  use another fingerprint to extract from $x$  (almost) pure randomness,
even conditional on the  information on $x$ held by  the eavesdropper. 
In the first case, we need to extract from $x$ all the information about this string that Bob misses
(with possibly a minor overhead, which contains a part of information that Bob already knows).
In the second case, we need to extract from $x$  the information that  the eavesdropper misses after observing Step 1
(possibly with a minor shortage, so that the randomness contained in $x$ remains partially unrevealed). 
Both ideas have been used in other secret key agreement protocols. However, their technical implementation  in the AIT framework requires specific constructions  needed for this setting,  such as, to give just one example, \emph{prefix extractors}, which were first  introduced explicitly  in~\cite{rareva:j:extractor}, and in AIT in~\cite{mus-rom-she:j:muchnik}. 
\end{remark}


\begin{proof}[Proof of Theorem~\ref{t:upper}]
Similarly to the argument in the \emph{light upper bound} of the previous section, 
we are going to combine  inequalities~\eqref{comm-inf-relativ-simplified} and~\eqref{eq-0}.  All we need to do is to establish inequality~\eqref{eq-0} for the case when the number of rounds  $k$ is non-constant.
The challenge is that we cannot iterate $k$  applications of Lemma~\ref{lemma-I(a:b:t)-is-positive}. If $k$ is not a constant, then
the sum of $k$ copies of $O(\log n)$ terms in \eqref{e:e3} becomes too large. 
Thus,  we need to handle the transcript $t$ ``in one piece," without splitting it in $k$ messages.
To this end we have to use instead of Lemma~\ref{lemma-I(a:b:t)-is-positive} some other tool.

We start with a reminder of the notion of a combinatorial rectangle (quite standard in communication complexity).
\begin{definition}
A set $S\subset \{0,1\}^* \times\{0,1\}^*$ is a \emph{combinatorial rectangle}, if it can be represented as a Cartesian product $S=A\times B$ for some $A,B\subset \{0,1\}^*$.
\end{definition}
For example, in Fig.~\ref{fig:comb-rectangle} the black cells in the table form a combinatorial rectangle (the black cells occupy the intersection of columns $a_1,a_2,a_3,a_4$ and rows $b_1,b_2,b_3$).
\begin{figure}[h]
\begin{center}
\includegraphics[scale=1.0]{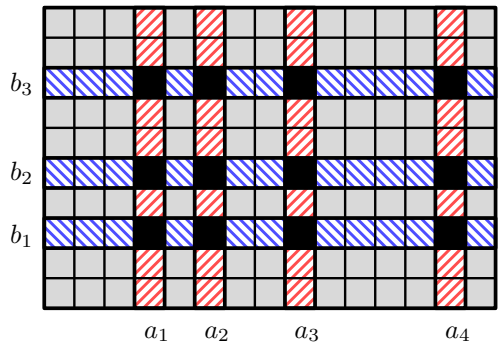}
\caption{Combinatorial rectangles.}\label{fig:comb-rectangle}
\end{center}
\end{figure}
We use the notion of a combinatorial rectangle in the following (less conventional) definition.
\begin{definition}
We say that a computable function $f : \{0,1\}^* \times \{0,1\}^* \to \{0,1\}^*$ is \emph{transcript-like}, if for every $t$, the pre-image $f^{-1}(t)$ is a combinatorial rectangle.
\end{definition}
For example, the cells in the table in Fig.~\ref{fig:comb-rectangle} are colored in four different colors (solid black, solid grey, and two types of hatching), and cells of each color form a combinatorial rectangle. Thus, the function that maps each cell to its color is transcript-like.

By definition, for every  transcript-like function $f$, if $f(a,b')=f(a',b)=t$ for some $a,a',b,b'$, then we have also $f(a,b)=t$.
Notice that for every deterministic communication protocol with two parties (Alice and Bob) the mapping
 \[
    f:  [\mbox{input of Alice $\xa$}] \times [\mbox{input of Bob $\xb$}] \mapsto [\mbox{the transcript of the protocol}]
\]
is transcript-like. A similar statement holds for randomized protocols with private source of randomness, we should only join  the random bits given to Alice and Bob with their inputs : the mapping
 \[
    f:  \langle \xa,\ra \rangle \times \langle \xb,\rb \rangle  \mapsto [\mbox{the transcript of the protocol}]
\]
is also transcript-like.

Now we formulate a generalization  of Lemma~\ref{lemma-I(a:b:t)-is-positive}~(a), which can work with the transcript of the protocol `in one piece'. 
\begin{lemma}\label{lemma-I(a:b:t)-is-positive-for-transcripts}
If  $f : \{0,1\}^* \times \{0,1\}^* \to \{0,1\}^*$ is transcript-like computable function,  then for all $a,b$
 \begin{equation} \label{I(a:b:t)-is-positive-for-transcripts}
  I(a:b \mid f(a,b))  \le^+ I(a:b).
\end{equation}
\end{lemma}
\begin{remark}
If $f$ depends only on $a$ or only on $b$, then $f$ is transcript-like. Thus, 
Lemma~\ref{lemma-I(a:b:t)-is-positive}~(a) is a special case of Lemma~\ref{lemma-I(a:b:t)-is-positive-for-transcripts}.
\end{remark}
\begin{proof}
Let us fix a pair  $(a,b)$ and denote $t:=f(a,b)$. We say that a string $a'$ is an A-\emph{clone} of $a$ (conditional on $t$), if  
 \begin{itemize}
 \item[(i${_A}$)] there exists another string $b'$ such that $f(a',b') = t$, and
 \item[(ii${_A}$)]  $C(a')\le C(a)$.
 \end{itemize}
Similarly, we say that $b'$ is a B-\emph{clone} of $b$ (conditional on $t$), if
 \begin{itemize}
 \item[(i${_B}$)] there exists another tuple $a'$ such that $f(a',b') = t$, and
 \item[(ii${_B}$)] $C(b')\le C(b)$.
 \end{itemize}
 \begin{remark}
The A- and B-clones are in some sense ``similar''  to $a$ and $b$ respectively. In what follows we show that the definition of clones is not too restrictive:  there exist  many clones of $a$ and $b$.
\end{remark}
 
Denote the sets of A-clones and B-clones by  $\mathrm{Clones}_A$ and $\mathrm{Clones}_B$ respectively. By definition we have $a\in \mathrm{Clones}_A$ and $b\in \mathrm{Clones}_B$. Let us show that in fact there exist many other A- and B-clones.

\begin{claim} 
If complexities of $a$ and $b$ are bounded by $n$, then
$|\mathrm{Clones}_A| \geq 2^{C(a \mid f(a,b))-O(\log n)} $ and $|\mathrm{Clones}_B| \geq 2^{C(b \mid f(a,b))-O(\log n)}  $. 
\label{claim-number-of-clones}
\end{claim}
\emph{Proof of the claim:} Notice that for every computable function $f$ the property
\[
\emph{a string }a'\emph{ is a clone of some }a \emph{ of complexity }k\emph{ conditional on }t
\]
is a semi-decidable predicate of the triple $\langle a', t, k\rangle$.
Therefore,  given $t$ and the value of  $C(a)$, we can run the process of enumeration of all A-clones of $a$. Thus, to identify some specific A-clone $a'$ of $a$  given $t$,  we need to know the  complexity of $a$ and the index of appearance of  this specific clone in the enumeration of all A-clones. Hence, $C(a'\mid t)$ is not greater than
  $ \log |\mathrm{Clones}_A|  + O(\log n). $
 In particular, this inequality is true for $a$ itself.
Similarly, $C(b\mid t)$ is not greater than
  $ \log |\mathrm{Clones}_B|  + O(\log n). $
Taking the exponent of both sides of these  inequalities,  we obtain the Claim.

\smallskip

Let us  take a pair of ``typical'' clones
 $
 (a',b')  \in \mathrm{Clones}_A \times  \mathrm{Clones}_B
$
 that maximizes $C(a',b' \mid t)$. 
Due to the claim above,  this maximum  is
 $$
 \begin{array}{rcl}
 C(a',b' \mid t)  &\eqp &  \log |\mathrm{Clones}_A\times   \mathrm{Clones}_B|\\
    & \ge^+&  C(a\mid t) + C(b \mid t) .
\end{array}
$$
Thus, we have
\begin{equation}\label{eq-1}
 C(a', b' , t) \eqp   C(t)  + C(a',b' \mid t) \ge^+  C(t) +  C(a\mid t) + C(b \mid t).
\end{equation}

By definition of clones, $a'$ and $b'$ belong to the projection of the pre-image $f^{-1}(t)$.  
Since $f$ is a transcript-like function, it follows that $f(a',b')=t$. That is, given $a'$ and $b'$ we can compute $t$, so $C(t \mid a',b' ) \eqp  0$.  
Hence, the left-hand side of \eqref{eq-1} can be upper bounded as
\begin{equation}\label{eq-2}
 C(a',b', t) \le ^+  C(a') + C(b') \le^+  C(a) + C(b)
\end{equation}
(in the last inequality we used again the definition of clones).  Combining \eqref{eq-1} and \eqref{eq-2} we obtain
\begin{equation}\label{eq-3}
 C(t)+ C(a \mid t) + C(b \mid t)  \le^+  C(a) + C(b).
\end{equation}
Given $C(t \mid a,b)\eqp  0$, inequality~\eqref{eq-3} rewrites to
$
 I(a:b \mid t )  \le^+ I(a:b),
$
 and lemma is proven.
\end{proof}

Now we are ready to prove the theorem.  Denote by $\pi $ the function
  \[
  \pi :    \langle \xa',\ra' \rangle \times \langle \xb',\rb' \rangle  \mapsto [\mbox{the transcript of the protocol}]
  \]
mapping the input data of Alice and Bob to the corresponding transcript of the communication protocol. Clearly, $\pi$ is a transcript-like function (as the outcome of any  communication protocol). 
By applying Lemma~\ref{lemma-I(a:b:t)-is-positive-for-transcripts} to the function $\pi$ we  get inequality~\eqref{eq-0}. This concludes to proof of the theorem. 
\end{proof}
\begin{remark}
The proof of inequality~\eqref{eq-0} applies even to ``invalid'' input data $\xa$ and $\xb$, when Alice and Bob are given false information on the complexity profile of the main input components $x$ and $y$. \textup(See Remark~\ref{remark-on-valid-inputs}\textup). 
\end{remark}
\begin{remark}
In some sense the proof of the weaker  `light upper bound' in section Section~\ref{s:warmup} is more robust than the proof of Theorem~\ref{t:upper},
since it applies to protocols with  small advice obtained from a trustable all-powerful helper, see Section~\ref{s:protocols-with-advice}.
\end{remark}

\section{Secret key agreement for $3$ or more parties}
 \label{section:3-sources}

We analyze  secret key agreement  for $3$ or more parties. All the results in this section   are valid for any constant number $\ell \geq 3$ of parties. For notational convenience, we handle the case of $\ell = 3$ parties, but the proofs with straightforward adaptations are valid for any constant $\ell$.

Thus, this time there are three parties, which we call Alice, Bob, and Charlie. Alice has a string $\xa$, Bob has a string $\xb$, and Charlie has a string $\xc$. They also have private random bits $\ra$, respectively $\rb$ and $\rc$. They run a $k$-round protocol. In each of the $k$ rounds, each party broadcasts a message to the other two parties, where the message is a string computed from the party's input string and private random bits, and the messages from the previous rounds. After the completion of the $k$ rounds, each party computes a string. The requirement is that with probability at least $1-\epsilon$, they compute the same string, and that this string is random conditioned by the transcript of the protocol.

Formally, a $k$-round  $3$-party  protocol for secret key agreement uses three computable functions $A, B, C$, and runs as follows. The first party has as input an $n$-bit string $\xa$ and uses private randomness $\ra$,  the second  party has as input an $n$-bit  string $\xb$ and uses private randomness $\rb$,  and the third  party has as input an $n$-bit string $\xc$ and uses private randomness $\rc$.   The protocol consists of the following calculations:
\[
\begin{array}{lll}
t_1 = A(\xa, \ra), & t_2 = B(\xb, \rb), &  t_3 = C(\xc, \rc), \\
t_4  = A(\xa, \ra, t[1:3]), & t_5  = B(\xb, \rb, t[1:3]),  & t_6  = C(\xc, \rc, t[1:3]), \\

\vdots \\

t_{3k-2} = A(\xa, \ra, t[1:3(k-1)]),  &  t_{3k-1} = B(\xb, \rb, t[1:3(k-1)]),  & t_{3k}=C(\xc, \rc,t[1:3(k-1)])  
\end{array}
\]
Each row corresponds to one round and shows the messages that are broadcast in that round, and we use the notation $t[i:j]$ to denote the tuple of messages $(t_i, \ldots, t_j)$. We also denote $t = t[1:3k]$, the entire transcript of the protocol.
The protocol  succeeds with probability error $\epsilon$ and randomness deficiency $\delta(n)$  on the $3$-tuple  input  $(\xa, \xb, \xc)$ if  with probability $(1-\epsilon)$ over $\ra, \rb, \rc$, 
\begin{equation}
\label{e:e4}
A(\xa,\ra, t) = B(\xb, \rb, t)  = C(\xc, \rc, t)  \stackrel{def.}{=} z,
\end{equation}
and
\begin{equation}
\label{e:e5}
 C (z \mid t) \geq |z| - \delta(n).
\end{equation}

\subsection{Omniscience: definitions and explicit expressions}
We  transfer to AIT the notion of \emph{omniscience} introduced in IT (see the discussion in Section~\ref{s:contrib}):
\begin{definition}
\label{d:slepwolf}
\begin{itemize}
\item[(1)]  For each triple of strings $(\xxa, \xxb, \xxc)$, we denote by $S(\xxa,\xxb,\xxc)$ the set of all triples of integers $(\na, \nb, \nc)$ that satisfy the following inequalities:
\begin{equation}\label{eq-linear-constraints-co}
\begin{array}{ccc}
\na \geq C(\xxa \mid \xxb, \xxc), &  \nb \geq C(\xxb \mid \xxa, \xxc), & \nc \geq C(\xxc \mid \xxa, \xxb), \\
\na + \nb \geq C(\xxa, \xxb \mid \xxc), & \na + \nc \geq C(\xxa, \xxc \mid \xxb), &
 \nb + \nc \geq C(\xxb, \xxc \mid \xxa).
 \end{array}
\end{equation}
The constraints defining $S(\xxa, \xxb, \xxc)$ will be referred to as the \emph{Slepian-Wolf} constraints.

\item[(2)] We define $\co(\xxa, \xxb, \xxc)$ to be the minimal value of $\na + \nb + \nc$ subject to $\na, \nb, \nc$ satisfying the Slepian-Wolf constraints. ($\co$ stands for communication for omniscience.)

\end{itemize}
\end{definition}
This definition has a straightforward (though rather cumbersome) generalization for the tuples of size $\ell>3$:
\begin{definition}
\label{d:slepwolf-multivar}
\begin{itemize}
\item[(1)]  For each tuple of strings $(x_1, x_2, \ldots,x_\ell)$, we denote by $S(x_1, x_2, \ldots,x_\ell)$ the set of all tuples of integers $(n_1,n_2,\ldots,n_\ell)$ that satisfy the following inequalities: for every splitting of the set of all indices into two  disjoint nonempty sets,
$ \{1,\ldots,\ell\} =  \mathcal{I} \cup \mathcal{J}$, we have
\begin{equation}\label{eq-linear-constraints-co-multivar}
\begin{array}{ccc}
 \sum\limits_{i\in \mathcal{I}} n_i  &\ge& C(  x_{\mathcal{I}} \mid x_{\mathcal{J}}).
 \end{array}
\end{equation}
Similar to the case of $\ell=3$, the constraints defining $S(x_1, x_2, \ldots,x_\ell)$ are referred to as the \emph{Slepian-Wolf} constraints.

\item[(2)] We define $\co(x_1, x_2, \ldots,x_\ell)$ to be the minimal value of $n_1 + \dots + n_\ell$ subject to $n_1, \ldots, n_\ell$ satisfying the Slepian-Wolf constraints.
\end{itemize}
\end{definition}
The value of $\co(\xxa, \xxb, \xxc)$ can be viewed as the solution of a problem of linear programming defined by constraints~\eqref{eq-linear-constraints-co}. This solution can be computed explicitly, as shown in the following proposition.
\begin{proposition}[{see Example~3 in \cite{csi-nar:j:seckey}}]
\label{p:co-value}
The value of $\co(\xxa, \xxb, \xxc)$  is equal to the maximum of the following four quantities
 \[
 \begin{array}{l}
 C(\xxa \mid \xxb,\xxc) + C(\xxb,\xxc \mid \xxa),\\ 
 C(\xxb \mid \xxa,\xxc) + C (\xxa,\xxc \mid \xxb),\\ 
 C(\xxc \mid \xxa,\xxb) + C (\xxa,\xxb \mid \xxc),\\ 
 \frac12\left[ C(\xxa,\xxb\mid \xxc) +  C(\xxa,\xxc\mid \xxb) + C(\xxb,\xxc\mid \xxa)   \right].
 \end{array}
 \]
Accordingly,  the difference $C(\xxa,\xxb,\xxc) - \co(\xxa, \xxb, \xxc)$  is equal to the minimum of the following four quantities:
\begin{equation}\label{eq-c-minus-co-formula}
\begin{array}{l}
I(\xxa:\xxb,\xxc),\\
 I(\xxb:\xxa,\xxc), \\ 
 I(\xxc:\xxa,\xxb), \\
 \frac12\left[ C(\xxa) +C(\xxb) + C(\xxc) -   C(\xxa,\xxb, \xxc)   \right].
 \end{array}
\end{equation}
\end{proposition} 
This result can be extended to the case $\ell>3$ as follows.
\begin{proposition}[\cite{chan2015multivariate}]
\label{p:co-value-multivar}
For every tuples of strings $(x_1,\ldots, x_\ell)$ the difference $C(x_1,\ldots, x_\ell) - \co(x_1,\ldots, x_\ell)$  is equal to the minimum of the quantities
\begin{equation}\label{eq-c-minus-co-formula-multivar}
\begin{array}{l}
 \frac1{s-1}\left[ C(x_{{\mathcal{J}}_1}) +  \ldots + C(x_{{\mathcal{J}}_s})   -   C(x_1,\ldots,x_\ell)   \right].
 \end{array}
\end{equation}
over all splittings $\{1,\ldots,\ell\}={\mathcal{J}}_1 \cup {\mathcal{J}}_2 \cup \ldots \cup {\mathcal{J}}_s$, where  ${\mathcal{J}}_1, \ldots, {\mathcal{J}}_s$ are nonempty and disjoint.
\end{proposition}
\begin{remark}
It is easy to see that the minimum of four values \eqref{eq-c-minus-co-formula} is a special case of the minimum 
of  \eqref{eq-c-minus-co-formula-multivar} \textup(for $\ell=3$\textup).
\end{remark}

In what follows, we show that there exists a protocol that on every input tuple $(\xa, \xb, \xc)$ produces with high probability a secret key of length $C(\xa, \xb, \xc) - \co(\xa, \xb, \xc) - O(\log n)$ (provided the parties have the complexity profile of the input tuple), and that no protocol can produce a secret key of length larger than $C(\xa, \xb, \xc) - \co(\xa, \xb, \xc) + O(\log n)$.
We prove the upper bound in Section~\ref{s:three-parties-non-constant} and the lower bound in Section~\ref{s:three-parties-lower-bound}.

 \subsection{Negative result for multi-party protocols}
 \label{s:three-parties-non-constant}

In this section we prove the upper bound on the length of the longest secret key that Alice, Bob, and Charlie can agree upon. A weak version of this bound (for the protocols with $O(1)$ rounds) can be obtained with the technique of classic information inequalities, very similar to the proof of an analogous statement in \cite{csi-nar:j:seckey} in the setting of IT. We only need to ``translate''  the information inequalities for Shannon's entropy from the proof in \cite{csi-nar:j:seckey} in homologous inequalities for Kolmogorov complexity, see Section~\ref{s:weak-upper-bound-multivar}.
However, this technique does not work for protocols with non-constant number of rounds. So, to prove our next theorem we use a quite different method, employing  the method of ``clones'' and constraint (non classic) information inequalities. This argument generalizes the proof of Theorem~\ref{t:upper}.

\begin{theorem}
\label{t:upperbdthree} Let us consider a  $3$-party  protocol  for secret key agreement with error probability $\epsilon$.  Let $(\xa, \xb, \xc)$ be a $3$-tuple of  $n$-bit strings on which the protocol succeeds. Let $z$ be the random variable which represents the secret key computed from the input $(\xa, \xb, \xc)$ and let $t$ be the transcript of the protocol that produces $z$. Then,  for sufficiently large $n$,  with probability $1-O(\epsilon)$, 
\[
C(z \mid t) \leq C(\xa,\xb,\xc) - \co(\xa, \xb, \xc)+ O(\log (n/\epsilon)),
\]
where the constants in the $O(\cdot)$ notation depend on the universal machine and the protocol, but not on $(\xa, \xb, \xc)$.
\end{theorem}

\begin{proof}
Assume that Alice, Bob, and Charlie agreed on a common key $z$. Denote by $t$ the transcript of the communication protocol. We need to show that $C(z\mid t)$ is not greater than the minimum of four quantities~\eqref{eq-c-minus-co-formula}.

\smallskip

\emph{Trivial case:  no communication nor randomness.} Similarly to the \emph{light upper bound} in Section~\ref{s:warmup}, we start with a naive question: 
on which  $z$, can Alice, Bob, and Charlie  agree without any communication?  The upper bound from  Section~\ref{s:warmup} applies to the case of three participants. Indeed,  inequality~\eqref{comm-inf} specializes to
$$
C(z) \le^+ C(z\mid \xa,\xb) + C(z\mid \xc) + I(\xa,\xb:\xc).
$$
Hence, if $z$ can be computed by Alice, Bob, and Charlie given each of the strings $\xa,\xb,\xc$, then
\begin{equation}\label{eq-ab:c}
C(z) \le^+ I(\xa,\xb:\xc).
\end{equation}
Similarly, we have 
\begin{equation}\label{eq-ac:b}
C(z) \le^+ I(\xa,\xc:\xb)
\end{equation}
and
\begin{equation}\label{eq-bc:a}
C(z) \le^+ I(\xb,\xc:\xa).
\end{equation}
In addition to the three mutual information quantities handled above, the maximum in~\eqref{eq-c-minus-co-formula} involves another  quantity, which is specific for protocols with three participants. We treat it separately:

\begin{lemma} If $z$ can be computed from each of the strings $\xa,\xb,\xc$, then
\begin{equation}\label{eq-j}
C(z) \le^+  \frac12\left( C(\xa)+ C(\xb) + C(\xc)- C(\xa,\xb,\xc) \right).
\end{equation}
\label{lemma-eq-j}
\end{lemma}
(See the proof in Section~\ref{s:technical}.)

The combination of \eqref{eq-ab:c}, \eqref{eq-ac:b}, \eqref{eq-bc:a}, \eqref{eq-j} implies
that if a string $z$ can be ``extracted''  from $\xa$, $\xb$, and $\xc$ (without communication),  then its complexity $C(z)$
is not greater than \eqref{eq-c-minus-co-formula}.

\smallskip

\emph{Randomized protocol without communication.}
We join $\xa,\xb$, and $\xc$ with  strings of random bits  $\ra$, $\rb$, $\rc$, and repeat the same arguments for the ``common information'' $z$, which is 
extracted from the triple $\langle \xa,\ra\rangle$, $\langle \xb,\rb\rangle$, and $\langle \xc,\rc\rangle$. We obtain that $C(z)$ is not greater than the minimum of 
\begin{equation}\label{eq-c-minus-co-formula-with-rand}
\begin{array}{l}
I(\xa,\ra : \xb,\rb,\xc,\rc) ,\\ 
I(\xb,\rb : \xa,\ra,\xc,\ra) ,\\ 
I(\xc,\rc : \xa,\ra,\xb,\rb) ,\\ 
 \frac12\left[ C(\xa,\ra) +C(\xb,\rb) + C(\xc,\rc) -   C(\xa,\ra,\xb, \rb, \xc,\rc)   \right].
 \end{array}
\end{equation}
Then we notice that for  randomly chosen $\ra$, $\rb$, and $\rc$, with probability $1-\epsilon$
 the difference between \eqref{eq-c-minus-co-formula-with-rand}  and \eqref{eq-c-minus-co-formula}
is only  $O(\log n/\epsilon)$.

\smallskip

\emph{The general case: protocols with any number of rounds.} The argument from the case of \emph{zero communication} examined above easily relativizes: 
if Alice, Bob, and Charlie agree (after some communication) on a common key $z$, and the transcript of the communication is $t$,
then $C(z\mid t)$ is not greater than the minimum of 
\begin{equation}\label{eq-c-minus-co-formula-with-rand-and-t}
\begin{array}{l}
I(\xa,\ra : \xb,\rb,\xc,\rc \mid t) , \\
I(\xb,\rb : \xa,\ra,\xc,\rc \mid t) , \\
I(\xc,\rc : \xa,\ra,\xb,\rb \mid t) , \\
 \frac12\left[ C(\xa,\ra \mid t) +C(\xb,\rb \mid t) + C(\xc,\rc \mid t) -   C(\xa,\ra,\xb, \rb, \xc,\rc \mid t)   \right].
 \end{array}
\end{equation}
Let us prove that \eqref{eq-c-minus-co-formula-with-rand-and-t} is not greater than \eqref{eq-c-minus-co-formula-with-rand}.

From the proof of Theorem~\ref{t:upper} we know that the mutual information between parties' data cannot increase when we relativize it with the transcript $t$. Thus, the first three quantities in  \eqref{eq-c-minus-co-formula-with-rand} cannot become greater when we add to these terms the transcript $t$  as a  condition.

It remains to prove a similar inequality between the information quantity 
 $$
 J(\xa,\xb,\xc) := \frac12\left[ C(\xa)+ C(\xb) + C(\xc)- C(\xa,\xb,\xc) \right]
 $$
and its relativized version
 $$
 J(\xa,\xb,\xc \mid t) := \frac12\left[ C(\xa \mid t )+ C(\xb\mid t) + C(\xc\mid t)- C(\xa,\xb,\xc \mid t) \right].
 $$
Notice that, in general (for an arbitrary $t$), the inequality $J(\xa,\xb,\xc \mid t) \le^+ J(\xa,\xb,\xc)$ can be  \emph{false}. So we need to employ the fact that $t$ is the transcript of a communication protocol. 
Here we reuse the technique from Theorem~\ref{t:upper}.
We  adapt the notion of a transcript-like function to the protocols with three parties.
\begin{definition}
A set $S\subset \{0,1\}^* \times\{0,1\}^* \times\{0,1\}^*$ is a \emph{combinatorial parallelepiped}, if it can be represented as a Cartesian product $S=A\times B \times  C$ for some $A,B,C \subset \{0,1\}^*$.
\end{definition}
\begin{definition}
We say that a computable function $f : \{0,1\}^* \times \{0,1\}^* \times \{0,1\}^* \to \{0,1\}^*$ is \emph{transcript-like}, if for every $t$, the pre-image $f^{-1}(t)$ is a combinatorial parallelepiped.
\end{definition}

Similarly to the proof of Theorem~\ref{t:upper}, we will use the trick with ``clones'' of inputs. 
Let $f$ be any transcript-like function of three arguments.
We  fix a triple of strings  $(a,b,c)$ (of length at most $n$) 
and denote $t:=f(a,b,c)$. We say that a string $a'$ is an A-\emph{clone} of $a$ (conditional on $t$), if 
 \begin{itemize}
 \item[(i)] there exist some strings $b',c'$ such that $f(a',b',c') = t$, and
 \item[(ii)]   $C(a') \le C(a)$. 
 \end{itemize}
Denote the sets of A-clones  by  $\mathrm{Clones}_A$. Similarly, we define the set of B-\emph{clones} of $b$  (denoted   $\mathrm{Clones}_B$) and  C-\emph{clones} of $c$   (denoted   $\mathrm{Clones}_C$), conditional on the fixed value $t=f(a,b,c)$.

\smallskip

\begin{claim} Assuming that complexities of the strings $a$, $b$, and $c$ are bounded by an integer $n$, we get
$|\mathrm{Clones}_A|  \geq 2^{C(a \mid f(a,b,c))-O(\log n)}  $,
$|\mathrm{Clones}_B| \geq 2^{C(b \mid f(a,b,c))-O(\log n)}  $, 
 and $ |\mathrm{Clones}_C| \geq 2^{C(c \mid f(a,b,c))-O(\log n)} $.
\end{claim}

\begin{proof} Similar to the proof of Claim~\ref{claim-number-of-clones}.
\end{proof}

\begin{lemma}\label{lemma-j(a,b,c|t)}
If  $f : \{0,1\}^* \times \{0,1\}^* \times  \{0,1\}^* \to \{0,1\}^*$ is transcript-like function,  then for all $a,b,c$
we have  $J(a,b,c\mid f(a,b,c)) \le J(a,b,c)$.
\end{lemma}
\begin{proof}
Fix $a,b,c$ and denote $t:=f(a,b,c)$. 
We choose a triple 
$$
 (a',b',c' )  \in \mathrm{Clones}_A \times  \mathrm{Clones}_B \times \mathrm{Clones}_C
$$
that maximizes $C(a',b',c' \mid t)$. Due to the Claim above  we have
 $$
 \begin{array}{rcl}
 C(a',b',c', t) & \eqp & C(t)  + \log | \mathrm{Clones}_A \times \mathrm{Clones}_B \times \mathrm{Clones}_C| \\
  &\ge ^+& C(t) + C(a\mid t) +C(b\mid t) +C(c\mid t).
 \end{array}
 $$
On the other hand, $t$ can be deterministically computed as $f(a',b',c')$, so
\[
 C(a',b',c', t) \le^+ C(a') +C(b') +C(c') \le^+ C(a) +C(b) +C(c).
\]
Combining the last two inequalities we obtain
$$
C(t) + C(a\mid t) +C(b \mid t) +C(c \mid t)  \le^+ C(a) +C(b) +C(c),
$$
which rewrites (assuming $C(t\mid a,b,c)\eqp 0$) to 
$$
 C(a\mid t) +C(b\mid t) +C(c\mid t)  - C(a,b,c\mid t) \le^+ C(a) +C(b) +C(c) - C(a,b,c).
$$
Hence, $J(a,b,c\mid t) \le J(a,b,c)$, and the lemma is proven.
\end{proof}

Now we  conclude the  proof of the theorem.  Denote by $\pi $ the function
  \[
  \pi :    \langle \xa',\ra' \rangle \times \langle \xb',\rb' \rangle  \times \langle \xc',\rc' \rangle  \mapsto [\mbox{the transcript of the protocol}]
  \]
mapping the input data of Alice,  Bob, and Charlie to the  transcript of the communication protocol. 
This is a transcript-like function (as the outcome of any  communication protocol with three parties in the model \emph{input in the hand}). So we can apply 
Lemma~\ref{lemma-j(a,b,c|t)} to this mapping $\pi$, and we are done.
\end{proof}

\begin{remark}
The inequality $J(a,b,c\mid t) \le^+ J(a,b,c)$ can be easily proven in case when $t$ is a deterministic function of one of the arguments $a$, $b$ or $c$. This observation \textup(plus the idea to split the transcript $t$ into separate messages\textup) essentially gives the proof of Theorem~\ref{t:upperbd} \textup(for a constant number of rounds\textup). 
\end{remark}

\paragraph{The case of $\ell>3$ parties.}
Let us sketch the proof of a general version of Theorem~\ref{t:upperbdthree} for any  constant number of parties $\ell\ge3$.
 To this end we should generalize the information quantities  $J(a,b,c)$,  $J(a,b,c\mid t)$ defined above. We need to consider the $s$-valent version of $J$ for all $s=2,\ldots,\ell$,
 \[
  J(x_1,x_2,\ldots,x_s) := \frac{1}{s-1}\big[  C(x_1)+\ldots + C(x_s) - C(x_1,\ldots,x_s)\big]
 \] 
and 
 \[
  J(x_1,x_2,\ldots,x_s \mid t ) := \frac{1}{s-1}\big[  C(x_1 \mid t )+\ldots + C(x_s\mid t ) - C(x_1,\ldots,x_s \mid t)\big].
 \] 
(Observe that these quantities appear in  \eqref{eq-c-minus-co-formula-multivar}.)

The notions of a transcript-like function can be extended to the functions with $\ell>3$ arguments.
The technique of ``clones'' permits to  generalize Lemma~\ref{lemma-j(a,b,c|t)} to this case. It follows that for all transcript-like function $f=f(x_1,x_2,\ldots,x_s)$ it holds
 \[
 J(x_1,x_2,\ldots,x_s \mid f(x_1,x_2,\ldots,x_s) ) \le  J(x_1,x_2,\ldots,x_s).
 \]
 Combining this observation with Proposition~\ref{p:co-value-multivar} we obtain the generalization of Theorem~\ref{t:upperbdthree} for any number of parties $\ell>3$. We omit the details of this proof \textup(which is a rather straightforward but cumbersome generalization of the proof of Theorem~\ref{t:upperbdthree}\textup).

 \begin{remark}
\label{remark-on-non-uniform-protocols-bis-bis}
Similarly to the case of two party protocols, the upper bounds for multi-party protocols can be extended to the setting with non-uniform communication protocols (see Remarks~\ref{remark-on-non-uniform-protocols},~\ref{remark-on-non-uniform-protocols-bis}). The generalization of the proof is quite straightforward, though  the statement of the theorem becomes more technical (we have to append a description of the protocol to the condition of all complexity terms involved in the theorem).
\end{remark}

\subsection{Positive results for  multi-party protocols}\label{s:three-parties-lower-bound}

In this section we prove the positive statement: we describe a 3-party  protocol that on every input tuple $(\xa, \xb, \xc)$ produces with high probability a secret key of length 
$$C(\xa, \xb, \xc) - \co(\xa, \xb, \xc) - O(\log n)$$ 
assuming that the parties are given the complexity profile of the input tuple.

\begin{theorem}
\label{t:multilower}
There exists a $3$-party protocol for secret key agreement with the following characteristics:
\begin{itemize}
\item[(1)] For every $n$, for every tuple $(\xxa, \xxb, \xxc)$ of $n$-bit strings, for every $\epsilon > 0$, if 
Alice's input $\xa$ consists of $\xxa$,  the complexity profile of the tuple and $\epsilon$, 
Bob's input  $\xb$ consists  of $\xxb$,  the complexity profile of the tuple and $\epsilon$,  and
Charlie' input $\xc$ consists of $\xxc$, the complexity profile of the tuple and $\epsilon$, 
then, at the end,  the three parties compute with probability $1 - O(\epsilon)$ a common string $z$ such that  $$C(z \mid t) \geq |z| - O(\log(1/\epsilon)) \mbox{ and } |z| \geq C(\xxa, \xxb, \xxc) - \co(\xxa, \xxb, \xxc) - O( \log (n/\epsilon)), $$ where $t$ is the transcript of the protocol, and the constants in the $O(\cdot)$ notation depend only on the universal machine.
\item[(2)] The protocol has $1$ round  and each party uses $O(\log(n/\epsilon))$ random bits.
\end{itemize}
\end{theorem}
\begin{proof}
We first give an overview of the protocol. 

In the first step, the three parties compute an optimal tuple $(\na, \nb, \nc)$ for the optimization problems  that defines  $\co(\xxa, \xxb, \xxc)$ (see Definition~\ref{d:slepwolf}), and then Alice computes $\pa$ of length $\na + O(\log n)$, Bob computes $\pb$ of length $\nb + O(\log n)$, and Charlie computes $\pc$  of length $\nc + O(\log n)$, such that with high probability $(\xa, \pb, \pc)$ is a program for $(\xxa, \xxb, \xxc)$,  $(\pa, \xb, \pc)$ is a program for $(\xxa, \xxb, \xxc)$,  and  $(\pa, \pb, \xc)$ is a program for $(\xxa, \xxb, \xxc)$. Next,  Alice broadcasts $\pa$, Bob broadcasts $\pb$, and Charlie broadcasts $\pc$. Alice also broadcasts a random string $r$ of length $O(\log n)$ bits, which will be used in the second step.

In the second step, each party  first computes $(\xxa, \xxb, \xxc)$. Next, each party, using the randomness $r$,  computes a string $z$ that with high probability is a short program for $(\xxa, \xxb, \xxc)$ given $(\pa, \pb, \pc)$.  We will show that $z$ has the desired properties.

We proceed with the details.

Step 1. Using the complexity profile of $(\xxa, \xxb, \xxc)$, each party computes a common tuple $(\na, \nb, \nc)$ which is optimal for the optimization problem that defines $\co(\xxa, \xxb, \xxc)$ (there may be more optimal tuples, but the parties agree on choosing  the same tuple in some canonical way).

Next,  Alice  uses the encoding procedure $E$ from Theorem~\ref{t:compressmulti} for $\ell=2$, on input $x_1, \na, \epsilon$ and obtains $\pa$.
The length of $\pa$ is bounded by $\na + O(\log (n/\epsilon))$. Similarly,  Bob and Charlie obtain $\pb$ and respectively $\pc$, from $x_2, \nb$, and respectively $x_3, \nc$. The parties broadcast  $\pa, \pb, \pc$ as indicated in the overview.
\smallskip

Step 2.  Alice uses the decoding algorithm $D$ from Theorem~\ref{t:compressmulti} for $\ell=2$,  on input $x_1, \pb, \pc$. Since $\nb, \nc$ satisfy the Slepian-Wolf constraints (\ie, $\nb + \nc \geq C(\xxb , \xxc \mid \xxa)$, $\nb \geq C(\xxb \mid \xxa, \xxc)$, $\nc \geq C(\xxc \mid \xxa, \xxb)$),  it follows that  with probability $1 - O(\epsilon)$,  Alice can reconstruct $\xxb$ and $\xxc$, and since she has $\xxa$, she obtains $(\xxa, \xxb, \xxc)$. The similar facts hold  for Bob and Charlie. Therefore, with probability $1 - O(\epsilon)$, each party obtains the tuple $(\xxa, \xxb, \xxc)$, an event which henceforth we assume to hold.

For the next operation, the plan is  to use the encoding procedure given in Theorem~\ref{t:compression} and compute a short program for $(\xxa, \xxb, \xxc)$ given $(\pa, \pb, \pc)$.  For that, the parties need to have a tight upper bound for $C(\xxa, \xxb, \xxc \mid \pa, \pb, \pc)$. Such an upper bound is provided in the following claim, which will be proved later.

\begin{claim}
\label{c:sestimate}
There exists a constant $d$ such that $$C(\xxa, \xxb, \xxc \mid \pa, \pb, \pc) \leq C(\xxa, \xxb, \xxc) - \co(\xxa, \xxb,\xxc) + d \log n.$$
\end{claim}

Now, using the upper bound provided by Claim~\ref{c:sestimate}, Alice (and also Bob, Charlie) use the encoding procedure from Theorem~\ref{t:compression} and on input $(\xxa, \xxb, \xxc)$, upper bound $k = C(\xxa, \xxb, \xxc) - \co(\xxa, \xxb, \xxc) + d \log n$ and randomness $r$, computes a string $z$, of length $|z| = k + c \log (n/\epsilon)$ for some constant $c$, that, with probability $1 -O(\epsilon)$ is a program for $(\xxa, \xxb, \xxc)$ given $(\pa, \pb, \pc)$. The transcript of the protocol is $t = (\pa, \pb, \pc, r)$.
We have
\[
\begin{array}{rl}
C(z \mid t) &= C(z \mid \pa, \pb, \pc,r) \eqp C(z \mid \pa, \pb, \pc) \quad \mbox{(because $|r| = O(\log n)$)} \\
& \geqp C(\xxa, \xxb, \xxc \mid \pa, \pb, \pc)  \quad \mbox{(because $z$ is a program for $(\xxa, \xxb, \xxc)$ given $\pa, \pb, \pc$)} \\
&\eqp C(\xxa, \xxb, \xxc) - C(\pa, \pb, \pc)  \\
& \quad  \quad \mbox{(because $\pa$ is computed from $\xxa$ and $O(\log n)$ bits, similarly for $\pb, \pc$)} \\
&\geqp C(\xxa, \xxb, \xxc) - (\na + \nb + \nc) \quad \mbox{( $|\pa| \leqp \na$, similarly for $\pb, \pc$)} \\
&= C(\xxa, \xxb, \xxc) - \co(\xxa, \xxb, \xxc) \\
&\geqp |z|,
\end{array}
\]
which shows that $z$ is a secret key with randomness deficiency $O(\log (n/\epsilon))$. Using randomness extractors and an additional random string $s$ broadcast by Alice in Step 1, similar to the proof of Theorem~\ref{t:lower}, the randomness deficiency of $z$ can be reduced to $O(\log(1/\epsilon))$.

It remains to prove Claim~\ref{c:sestimate}.
\smallskip

\begin{proof}[Proof of Claim~\ref{c:sestimate}]
Suppose that Alice, Bob, and Charlie  know $s =  C(\xxa, \xxb, \xxc \mid \pa, \pb, \pc)$. Since each of $\xxa, \xxb$ and $\xxc$ are $n$-bits long,  $C(\xxa, \xxb, \xxc \mid \pa, \pb, \pc) \leq 3n + O(1)$, and thus the length of $s$ in binary representation is bounded by $\log n + O(1)$ bits. Suppose $s$ is given as part of the input, \ie, 
Alice's input is $(\xxa, s)$ instead of just $\xxa$, and similarly $s$ is given to Bob and Charlie.
Then Alice, Bob and Charlie can use the encoding procedure given in Theorem~\ref{t:compression}: From $(\xxa, \xxb, \xxc)$, the upper bound $s$  and the randomness $r$, they compute a string $z'$ that, with probability $1- 1/n$ is a program for $(\xxa, \xxb, \xxc)$ given $(\pa, \pb, \pc)$ and $z'$ has length
$|z'| \leq s + c \log n = C(\xxa, \xxb, \xxc \mid \pa, \pb, \pc) + c \log n$ for some constant $c$. We have
\[
\begin{array}{rl}
C(z' \mid \pa, \pb, \pc) & \geq C(\xxa, \xxb, \xxc \mid \pa, \pb, \pc) \\
& \mbox{\quad \quad \quad \quad (because $z'$ is a program for $(\xxa,\xxb,\xxc)$ given $(\pa, \pb, \pc)$} \\
				& \geq |z'| - c \log n.
\end{array}
\]
The transcript of the protocol (that uses $s$ as part of the input) is $t = (\pa, \pb, \pc, r)$ and since the length of $r$ is only $O(\log n)$, it follows that $C(z' \mid t) \eqp C(z' \mid \pa, \pb, \pc)$ and therefore $C(z' \mid t) \geqp |z'|$. It follows that $z'$ is a secret key that Alice, Bob, and Charlie have obtained via the protocol with $s$ added to the input.  It follows from 
Theorem~\ref{t:upperbdthree} that we get the upper bound
$C(z' \mid t) \leq C(\xxa, \xxb, \xxc,s) - \co(\xxa, \xxb, \xxc) + d' \log n$, for some constant $d'$. 
Since $s$ is only
$\log n + O(1)$ bits long, we can drop $s$ at the cost of increasing the constant $d'$.
The conclusion follows, because $C(\xxa, \xxb, \xxc \mid \pa, \pb, \pc)\leq C(z' \mid \pa, \pb,\pc) \leq C(z' \mid t)$.
\end{proof} \emph{(End of proof of Claim~\ref{c:sestimate}.)}
\end{proof}

\section{Communication complexity of  secret key agreement}
\label{s:comm}
It is of interest to find the communication complexity for the task of finding a shared secret key  having the optimal length of $I(x:y)$.   We solve this problem in the model of randomized protocols with \emph{public random bits}, visible to Alice, Bob, and the adversary.  This model is obtained by modifying slightly the definition from Section~\ref{s:defprotocol} (in which the random bits are private): we require that $\ra = \rb = r$ and we change equation~\eqref{e:e2} to $C(z \mid t, r) \geq |z| - \delta(n)$.

The protocol presented in the proof of Theorem~\ref{t:lower} solves the task with communication $\min (C(x \mid y), C(y \mid x)) + O(\log n)$.
As presented, that protocol uses private random bits, but it can be easily modified to work in the  model with  public randomness.  Indeed,  Alice and Bob use only $O(\log n)$ private  random bits. So, if we make these bits known to the adversary, this will not affect much the secrecy of the common key $z$ produced in the protocol. That is,  we still have the inequality
 \[
  C(z \mid [\text{everything known to the adversary}]) \ge |z| - O(\log n),
 \]
only the constant in the term $O(\log n)$ becomes bigger (compared to equation~\eqref{e:tz}). 
Thus, the assumption that all participants of the protocol have access to a public source of random bits, only makes the protocol simpler, and with decreased communication complexity by  $O(\log n)$. 

We argue that within the model with public randomness the communication complexity of this protocol  is optimal, up to the $O(\log n)$ term.
In what follows, we assume as usual that  Alice is given a string $x$ and Bob is given a string $y$, and both parties know the complexity profile of $(x,y)$.

\begin{theorem}
\label{thm:comm-complexity} Let $\epsilon, \delta_1, \delta_2$ be arbitrary positive real constants.
There is no secret key agreement protocol with public random  bits such that for all inputs $x$ and $y$,  
\begin{enumerate}
\item  the communication complexity of the protocol (the total number of all bits sent by Alice and Bob) is less than  $(1-\delta_1)\min\{  C(x \mid y), C(y \mid x)   \}$, 
\item  Alice and Bob agree with probability $>\epsilon$  on a  common key $z$ such that
$C(z\mid t,r) > \delta_2  I(x:y)$,
where $r$ is the public randomness and $t=t(x,y,r)$ is the transcript of the protocol.
\end{enumerate}
\end{theorem}

\emph{Specifying the parameters.}   To prove the theorem, we find a `hard' pair $(x,y)$, on which any communication protocol with public randomness fails.  
More specifically, we will choose a `hard' pair with the profile
\begin{equation}\label{standard-profile}
C(x)\eqp 2n,\  C(y)\eqp 2n,\  C(x,y)\eqp 3n
\end{equation}
(for large enough $n$);  notice that \eqref{standard-profile} implies $I(x:y)\eqp n$.
We will show that if the protocol satisfies condition~(1) from Theorem~\ref{thm:comm-complexity}, then, for the chosen `hard' strings $(x,y)$,  condition~(2) from this theorems is false.

\begin{remark} 
The construction that we  present in this section can be adapted to construct a `hard'  pair $(x,y)$ with any given nontrivial complexity profiles  (we only need that $C(x|y)$ and $C(y|x)$ are both not too small).
We could make Theorem~\ref{thm:comm-complexity} stronger by replacing the constant values $\delta_i$ 
by $\delta_i = \lambda_i\log n/n$   with  large enough constant coefficients $\lambda_i>0$ ($i=1,2$).
We sacrifice the generality  to simplify the notation.
\end{remark}

\emph{Digression: common information.} Before we prove Theorem~\ref{thm:comm-complexity}, we remind some results on common information.
We say (somewhat informally) that strings  $x$ and $y$ share $k$ bits of common information, if there exists a string $w$ of complexity $k$ such that $w$ can be easily ``extracted'' from $x$ and from $y$, i.e., $C(w \mid x)\approx 0$ and $C(w \mid y)\approx0$. To make this statement formal, we must specify the precision of the equalities $C(w \mid x)\approx 0$ and $C(w \mid y)\approx0$. In the next definition we introduce a suitable notation.

\begin{definition}\label{d:common-i}
We say that $k$ bits of common information can be extracted from $x$ and $y$ with inaccuracies $(k_A,k_B)$, if there exists a string $w$ such that 
 $$
 \left\{
 \begin{array}{ccl}
 C(w) &\ge & k,\\
 C(w\mid x) &\le & k_A,\\
 C(w \mid y) &\le & k_B.
 \end{array}
 \right.
 $$
 Denote  by $\ci(x,y)$  the  set of all triples $(k,k_A,k_B)$  such that $k$ bits of common information can be extracted from $x$ and $y$ with inaccuracies $(k_A,k_B)$.
\end{definition}

It is know that  $\ci(x,y)$ is not uniquely defined by the complexity profile of $(x,y)$. That is,  pairs $(x,y)$ and $(x',y')$ with very similar (or even identical) complexity profiles can have very different properties of ``extractability'' of the mutual information, 
see \cite{che-muc-rom-she-ver:j:commoninfo}.
The next theorem describes the  minimal possible set of triples $\ci(x,y)$ for a pair with complexity profile \eqref{standard-profile}.

\begin{theorem} \cite{muc:j:commoninfo} \label{thm:muchnik}
For every $\delta>0$ and for all large enough $n$ there exists a pair of strings $(x,y)$ such that
$C(x)\eqp 2n$, $C(y)\eqp 2n$, $C(x,y)\eqp  3n$, and \emph{no} triple $(k,k_A,k_B)$ satisfying
 $$
 \left\{
 \begin{array}{l}
 k_A \le  (1-\delta)n,\\
 k_B \le  (1-\delta)n,\\
 k_A+k_B +\delta n \le  k
 \end{array}
 \right.
 $$
belongs to $\ci(x,y)$.
\end{theorem}

I.~Razenshteyn proposed a more constructive version of Theorem~\ref{thm:muchnik}. Though we cannot  algorithmically construct a pair $(x,y)$ satisfying Theorem~\ref{thm:muchnik}  (there is obviously no short algorithmic description of any object of high Kolmogorov complexity), we can immerse such a pair in a large constructive set where the majority of elements have the desired property:

\begin{theorem} \cite{razenshteyn-2011}\label{thm:razenshteyn}
For every $\delta>0$ there exists an algorithm that for all large enough $n$ generates some list of pairs $S_n$ such that 
for the majority of $(x,y)\in S_n$ 
$C(x)\eqp 2n$, $C(y)\eqp 2n$, $C(x,y)\eqp  3n$, and \emph{no} triple $(k,k_A,k_B)$ satisfying
 $$
 \left\{
 \begin{array}{l}
 k_A \le  (1-\delta)n,\\
 k_B \le  (1-\delta)n,\\
 k_A+k_B +\delta n \le  k
 \end{array}
 \right.
 $$
belongs to $\ci(x,y)$.
\end{theorem}

\begin{remark} Technically, \cite{razenshteyn-2011} formulates  Theorem~\ref{thm:razenshteyn} only in the symmetric setting $k_A=k_B$. However Razenshteyn's proof applies in the general case.
\end{remark}

\begin{remark}
When we claim in Theorem~\ref{thm:razenshteyn} that an algorithm \emph{generates} a set $S_n$, we mean that on the input $n$ the algorithm prints the list of elements of this set and stops. It is important that the enumeration of $S_n$ has a distinguishable completion. 

 The weak version of ``constructiveness'' that appear in Theorem~\ref{thm:razenshteyn} is called \emph{stochasticity}, \cite{shen1983concept}. An individual string $x$ is called $(\alpha,\beta)$-stochastic, if there exists an algorithm of complexity $<\alpha $ that prints the list of elements of a set $S$ such that $x\in S$ and $C(x)\ge \log |S|-\beta$. That is, if $\alpha$ and $\beta$ are small, then an $(\alpha,\beta)$-stochastic string $x$ is, so to say, a \emph{typical} element in a \emph{simple} set. 

The notion of $(\alpha,\beta)$-stochasticity   can be easily adjusted to deal with pairs of strings. Thus, Theorem~\ref{thm:razenshteyn} asserts that there exist $(O(\log n), O(\log n))$-stochastic pairs with  the worst possible extractability of the mutual information.
\end{remark}

\emph{On the uselessness of random oracles.}
In computability theory,  it is very common to study ``relativized'' computations, i.e.,  computations with an oracle. In the theory of Kolmogorov complexity,  a natural version of relativization  consists in adding  an oracle (which can be a finite or infinite sequence of bits) to the ``condition'' part  of all complexity terms. So, the property of extracting of $k$ bits of  common information from $x$ and $y$ with an oracle $r$ can be formulated as follows:  there exists a string $w$ such that $C(w\mid r) = k$, while $C(w\mid x,r)\approx0$ and $C(w\mid y,r)\approx 0$. More generally, we define $\ci(x,y\mid r)$ as the set of all triples $(k,k_A, k_B)$ such that
 $$
 \left\{
 \begin{array}{ccl}
 C(w \mid r) &\ge & k,\\
 C(w \mid x,r) &\le & k_A,\\
 C(w \mid y,r) &\le & k_B.
 \end{array}
 \right.
 $$

If a string $r$ is independent of $(x,y)$, i.e., $I(x,y:r)\approx 0$, then it seems natural to conjecture that all reasonable properties of Kolmogorov complexity concerning $(x,y)$ do not change 
 if we relativize them with  $r$.  
 In particular, it seems plausible that 
 the difference between the sets $\ci(x,y|r)$ and $\ci(x,y)$ must be negligible. Surprisingly, this conjecture remains unproven. However, a version of this conjecture for stochastic pairs is known to be true.
 \begin{theorem} \cite{muchnik-romash} \label{thm:muchnik-romash} 
 If  $(x,y)$ is $(O(\log n),O(\log n))$-stochastic, then the set $\ci(x,y\mid r)$ belongs to an $O(\log n)$-neighborhood of $\ci(x,y)$ and vice-versa.
\end{theorem}
Now we are ready to prove Theorem~\ref{thm:comm-complexity}.

\begin{proof}[Proof  of Theorem~\ref{thm:comm-complexity}]
We fix some $\epsilon,\delta_1,\delta_2>0$. Assume for the sake of the argument that in a protocol with communication complexity $<(1-\delta_1)n$,  Alice and Bob agree with a probability greater than $\epsilon$  on a common  secret key $z$ such that $C(z\mid t,r)>\delta_2 n$. We show that this assumption leads to a contradiction. We will provide a counterexample, i.e., a pair $(x,y)$ on which the protocol fails. First, we consider a degenerate case --- a communication protocol without randomness.

\emph{(a) The case of deterministic protocols.} At first, we focus on a very special case: we assume that the communication protocol is \emph{deterministic} (Alice and Bob do not use  the source of random bits). In this case $C(z\mid t,r)>\delta_2 n$ simplifies to $C(z\mid t)>\delta_2 n$.

  Denote $w:=\langle t,z\rangle$. We look at this $w$ as sort of ``common information'' extracted from $x$ and $y$. Let us estimate the parameters of this ``shared part'' of two strings.
 \begin{itemize}
 \item {} [a trivial calculation] 
 $C(w \mid x) \le^+ C(t)+ C(z \mid t,x) \le^+ (1-\delta_1)n $,
 \item {} [a similar calculation] 
 $C(w \mid y) \le^+ C(t)+ C(z \mid t,y) \le^+ (1-\delta_1)n $, 
 \item {} [a less trivial but also straightforward calculation] 
 \[
 \begin{array}{rcll} 
    C(w)  &\eqp &  C(w \mid x)  +  C(w\mid y)  + I(x:y ) - I(x:y \mid w) & [\text{Section~\ref{s:technical}, Lemma~\ref{l:calculations}(a)}]\\
      & \eqp  &  C(w \mid x)  +  C(w\mid y)  +  I(x:y ) - I(x:y \mid t)  + C(z\mid t)  & [\text{Section~\ref{s:technical},  Lemma~\ref{l:calculations}(b)}]\\
    & \ge^+ &  C(w \mid x)  +  C(w\mid y)   + C(z\mid t) \\
     & \ge^+  &  C(w \mid x)  +  C(w \mid y) + \delta_2 n. 
   \end{array}
   \]      
\end{itemize}
For the transition from the second to the third line we used the inequality
$I(x:y ) - I(x:y \mid t)\ge^+0$,
which can be false for arbitrary $x,y,t$.
Fortunately, it is true if $t$  is the transcript of a communication protocol with inputs $x$ and $y$ (Lemma~\ref{lemma-I(a:b:t)-is-positive-for-transcripts}).
\smallskip

We know from Theorem~\ref{thm:muchnik}  that the three inequalities above cannot hold together, if $(x,y)$ is a pair with a minimal $\ci(x,y)$. Thus, we get a contradiction.

\emph{(b) Randomized protocols.} Now we want to extend the result from the previous paragraph to randomized protocol. 
We assume now that there exists a communication protocol with public random bits $r$ such that  for most $r$ the protocol produces 
a transcript $t$ of size $<(1-\delta_1)n$ that permits to extract from $x$ and $y$ a common key $z$ such that $C(z \mid t,r)\ge \delta_2 n$. 

Denote again $w:=\langle t,z\rangle$. Then the triple $(C(w\mid r), C(w\mid x,r), C(w \mid y,r))$ belongs to $\ci(x,y \mid r)$. 
We can repeat the calculations from Case (a),  adding the string $r$  to the condition of all complexities:
 \begin{itemize}
 \item {}  $C(w\mid x,r) \le^+ (1-\delta_1)n $,
 \item {}  $C(w \mid y,r) \le^+ (1-\delta_1)n $, 
 \item {}  $C(w \mid r)    \ge^+   C(w \mid x,r)  +  C(w \mid y,r) + \delta_2 n$
  \end{itemize}
We may assume that $r$ is independent of $(x,y)$,  i.e., $I(x,y:r) \eqp 0$. If $(x,y)$ is a stochastic pair, then from Theorem~\ref{thm:muchnik-romash} it follows that an almost the same triple  of integers
  $$
  (C(w \mid r)\pm O(\log n), C(w \mid x,r)\pm O(\log n), C(w \mid y,r)\pm O(\log n))
  $$
must belong to $\ci(x,y)$. Due to Theorem~\ref{thm:razenshteyn} this is false for some stochastic $(x,y)$, and we get a contradiction.
\end{proof}
\section{Several technical lemmas}
\label{s:technical}
\subsection{Proof of Lemma~\ref{l:def-reduction}}
\label{s:strongext}


Let us fix a  $\tz$ and $t$ that satisfy  the relations~\eqref{e:tz}.   We introduce some notation: $$n' = |\tz|, \quad \tx = (t, n', \epsilon), \quad \alpha = C(\tz \mid \tx).$$

By the relations~\eqref{e:tz} and the fact that $n' \leq n$, we have, for some constant $c$,
\[
C(\tz \mid \tx) \geq |\tz| - c \log (n/\epsilon).
\]
Let $k =   |\tz| - c \log (n/\epsilon)$. Thus, $C(\tz \mid \tx) \geq k$.

\if01
We use strong extractors. Recall that a function  $E: \zo^n \times \zo^d \mapping \zo^m$ is a $(k, \epsilon)$ extractor if for any set $B \subseteq \zo^n$ of size $|B| \geq 2^k$, and for any $A \subseteq \zo^m$,
\[
\prob(E(U_B, U_{\zo^d}) \in A) \in \big(\frac{|A|}{M}- \epsilon, \frac{|A|}{M} + \epsilon \big),
\]
where $U_B, U_{\zo^d}$ denote independent random variables uniformly distributed  on $B$, and respectively $\zo^d$, and $M = 2^m$.

A function  $E: \zo^n \times \zo^d \mapping \zo^m$ is a $(k, \epsilon)$ \emph{strong} extractor if the function
$E_1$ defined by $E_1(x,y) = (y, E(x,y))$ is a $(k, \epsilon)$ extractor. Using the probabilistic method, we can obtain strong extractors with $m = k - 2 \log(1/\epsilon) - O(1)$ and $d = \log n  + 2 \log (1/\epsilon) + O(1)$~\cite[Theorem 6.17]{vad:b:pseudorand}.
\fi 

We use $E: \zo^{n'} \times \zo^d \mapping \zo^m$, a $(k- c_1, \epsilon)$- strong extractor (for some constant $c_1$ that will be specified later), having $m = k - c_1 - 2 \log(1/\epsilon) - O(1)$ and $d = \log n' + 2 \log (1/\epsilon) + O(1)$.

In the following lemmas, $s$ is a seed of the above strong extractor $E$  chosen uniformly at random  and  $z = E(\tz, s)$.
\begin{lemma}
\label{l:extract1}
With probability $1- O(\epsilon)$ over the choice of $s$,
\[
C(\tz,s \mid \tx, \alpha) \leq C(z,s \mid \tx) + C(\tz \mid \tx) - m + \log(1/\epsilon) + O(1).
\]
\end{lemma}
\begin{lemma}
\label{l:extract2}
With probability $1- O(\epsilon)$ over the choice of $s$,  
\[
C(\tz,s \mid \tx, \alpha) \geq C(\tz \mid \tx, \alpha) + |s| - O(\log(1/\epsilon)) - O(1).
\]
\end{lemma}
\begin{lemma}
\label{l:extract3}
$C(\tz \mid \tx, \alpha) \geq  C(\tz \mid \tx) - O(1)$.
\end{lemma}
\begin{lemma}
\label{l:extract4}
For every $s$,  
$C(z,s \mid  \tx) \leq  |s| + C(z \mid s, \tx) + O(1)$.
\end{lemma}

Combining the four  lemmas, we obtain that, with probability $1- O(\epsilon)$, $C(z \mid s, \tx) \geq m - O(\log(1/\epsilon)) - O(1)$.  Since $|z| = m = I(x:y) - O(\log n) - O(\log(1/\epsilon))$ and taking into account that $\tx$ includes $t$, 
we obtain $C(z \mid  t, s) \geq |z| -  O(\log(1/\epsilon)) - O(1)$, and  that the length of $z$ is $I(x:y) - O(\log n) - O(1/\epsilon)$.

It remains to prove the lemmas.
\medskip

\begin{proof}[Proof of Lemma~\ref{l:extract1}] Recall that the function $E_1 : \zo^{n'} \times \zo^d \mapping \zo^d \times \zo^m$ defined by $E_1(w,s) = (s, E(w,s))$ is a $(k-c_1, \epsilon)$ extractor.
We view $E_1$ also as a bipartite graph, where the set of left nodes is $\zo^{n'}$, the set of right nodes is $\zo^d \times \zo^m$, and $(w, (s,z))$ is an edge if $E(w,s) = z$ (or, equivalently, $E_1(w,s) = (s,z)$). 
Let
\[
B = \{w \in \zo^{n'} \mid C(w \mid \tx) \leq C(\tz \mid \tx\}.
\]
We say that a right node $(s,z)$ in the above graph is \emph{heavy}, if its $B$-degree is $\geq (1/\epsilon)\frac{|B|}{M}$, where $M = 2^m$
 (the $B$-degree of a node is the number of edges coming from $B$ into the node).  Let $\rm{HEAVY}$ be the set of nodes that are heavy. By counting  from the left and also from the right
the edges that go out  from $B$, we get
\[
|\rm{HEAVY}| \cdot (1/\epsilon) \cdot \frac{|B|}{M} \leq |B| \cdot D,
\]
and therefore,
\[
\frac{|\rm{HEAVY}|}{M\cdot D} \leq \epsilon.
\]
A left node $w \in B$ is \emph{poor} if for more than $2 \epsilon$ fraction of $s \in \zo^d$, $E_1 (w,s)$ is heavy. 
\begin{claim}
\label{c:count}
The number of poor nodes is less than $2^{k-c_1}$.
\end{claim}
\begin{proof}
Let $\rm{POOR}$ be the set of poor nodes.
Since
\[
\begin{array}{rl}
\prob(E_1 (U_{\rm POOR}, U_{\zo^d}) \in {\rm HEAVY}) & > 2 \epsilon \\
& \geq  \frac{|\rm{HEAVY}|}{M\cdot D}  + \epsilon,
\end{array}
\]
the number of poor nodes has to be less than $2^{k-c_1}$ because otherwise the extractor property of $E_1$ would be contradicted.
\end{proof}
\begin{claim}
\label{c:poor}
$\tz$ is not poor.
\end{claim}
\begin{proof}
 Suppose $\tz$ is poor.  Then $C(\tz \mid \tx, \alpha) \leq k - c_1$, because, given $\tx, \alpha$, the string $tz$ is described by its index in a canonical enumeration of $B$
and the size of $B$ is bounded by $2^{k-c_1}$ (by Claim~\ref{c:count}). It is  known that for any $z', t'$, $C(z' \mid t', C(z' \mid t')) \geq C(z' \mid t') - O(1)$ (this is the relativized version of the known inequality $C(z' \mid C(z')) \geq C(z') - O(1)$, see~\cite[Exercise 44]{suv:b:kolmenglish}). It follows that
\begin{equation}
\label{e:cc}
C(\tz \mid \tx) \leq C(\tz \mid \tx, \alpha) +O(1) \leq k - c_1 +O(1).
\end{equation}
On the other hand, $C(\tz \mid \tx) \geq k$. If $c_1$ is large enough, we get a contradiction.
\end{proof}
\smallskip

Claim~\ref{c:poor} says that for a fraction of $1-2\epsilon$ of $s \in \zo^d$, the node $(s,p)$ is non-heavy, where $z = E(\tz,s)$. If $(s,z)$ is non-heavy, then the number of its left neighbors in $B$   is bounded by
\[
(1/\epsilon) \frac{|B|}{M} \leq (1/\epsilon) \frac{2^{C(\tz \mid \tx) +1}}{M} = 2^{C(\tz \mid \tx) - m + \log(1/\epsilon) + 1}.
\]
 It follows that, given $\tx, \alpha$, $\tz$ can be described by a non-heavy $(s,z)$ and its index in an enumeration of the left nodes of $(s,z)$ in $B$.
Thus, with probability $1-2\epsilon$ over the choice $s$,
\[
C(\tz, s \mid \tx, \alpha) \leq C(z,s \mid \tx) + (C(\tz \mid \tx) - m + \log(1/\epsilon) + 1) + O(1).
\]
This concludes the proof of the lemma.
\end{proof}

\begin{proof}[Proof of Lemma~\ref{l:extract2}]
Let $w = C(\tz \mid \tx, \alpha) + |s| - C(\tz, s \mid \tx, \alpha)$.  Note that $w$ is a random variable because it depends on $s$. We need to show that $w = O(\log (1/\epsilon))$, with probability $1- \epsilon$. 

By the standard counting argument, we have that with probability $1-\epsilon$, 
\begin{equation}
\label{es1}
C(s \mid \tx, \alpha) \geq |s| - \log(1/\epsilon).
\end{equation}
Let us fix an $s$ that satisfies inequality~\eqref{es1}. We show that the $w$ that corresponds to this $s$ satisfies $w= O(\log (1/\epsilon))$, from which the conclusion follows.

 If $w \leq 0$, we are done. So let us suppose that $w > 0$.

Let $u = C(\tz, s \mid \tx, \alpha) = C( \tz \mid \tx, \alpha) + |s| - w = (\alpha - O(1)) + |s| - w$ (the last equality holds because $C(\tz \mid \tx, \alpha) = C(\tz \mid \tx) - O(1)$, by the first part of inequality~\eqref{e:cc}).

 Let $A = \{(z', s') \mid C(z',s' \mid \tx, \alpha) \leq u\}$. Note that $|A| \leq 2^{u+1}$.  For every $z'$, we define $A_{z'} = \{s' \mid (z',s') \in A\}$. Let $e$ be the integer for which
$2^{e-1} < |A_{\tz}| \leq 2^e$.   Since $s \in A_{\tz}$ and $A_{\tz}$ can be enumerated given $w, \tx, \alpha$ and $O(1)$ bits, it follows that
\begin{equation}
\label{es2}
C(s \mid \tx, \alpha) \leq e + 2 \log w + O(1).
\end{equation}
Let $H = \{ z' \mid |A_{\tz}| > 2^{e-1} \}$.  Note that $|H| \leq |A| / 2^{e-1} \leq 2^{u-e+2}$ and
that $\tz$ is in $H$. We can write the index of $\tz$ in an enumeration of $H$ on exactly $u-e+2$ bits and taking into account that $u$ can be computed from $\tz, \alpha, w$ and $O(1)$ bits, it follows that $e$ can be derived as well from the index of $\tz$ and the information mentioned above. In this way, we have all the information needed to enumerate $H$ and we obtain
\begin{equation}
\label{es3}
C(\tz \mid \tx, \alpha) \leq u-e+2 + 2 \log w + O(1).
\end{equation}
Adding inequalities~\eqref{es2} and~\eqref{es3}, we get
\[
\begin{array}{rl}
C(\tz \mid \tx, \alpha) + C(s \mid \tx, \alpha) &\leq u + 4 \log w + O(1) \\
& \leq (C(\tz \mid \tx, \alpha) + |s| - w) + 4 \log w + O(1).
\end{array}
\]
Taking into account inequality~\eqref{es1}, we obtain $w - 4 \log w \leq \log(1/\epsilon) + O(1)$, which implies $w = O(\log(1/\epsilon))$, as desired.
\end{proof}

\begin{proof}[Proof of Lemma~\ref{l:extract3}]  This is the first part in the chain of inequalities~\eqref{e:cc}.  
\end{proof}

\begin{proof}[Proof of Lemma~\ref{l:extract4}] This follows from the fact that given $\tx$ (which contains $n'$) we can find the length of $s$, and therefore we do not  need delimiters to concatenate $s$ and a description of $z$ given $s$  and $\tx$.
\end{proof}

\subsection{Useful information inequalities}


\begin{lemma}
\label{l:chainrandom}
Let $f$ be a computable function, and suppose that for every string $x$, we choose a string $r$ uniformly at random among the strings of length $f(|x|)$.
Then, 

(1) with probability $1 - \epsilon$,\[
\big| C(x,r) - (C(x) + C(r \mid x)) \big| \leq O(\log (|x|/\epsilon)),
\]
(2) with probability $1 - \epsilon$,
\[
\big| C(x,r) - (C(r) + C(x \mid r)) \big| \leq O(\log (|x|/\epsilon))
\]
and 

(3) with probability $1 - \epsilon$,  
\[
C(x,r) = C(x) + |r| \pm O(\log (|x|/\epsilon)).
\]

\end{lemma}
\begin{proof}    
Note that part (3) follows, after a rescaling of $\epsilon$, from part (1) and the fact that with probability $1-\epsilon$, $C(r \mid x) \geq |r| - \log(1/\epsilon)$.

We prove (1). The proof of (2) is similar.

Clearly, $C(x,r) \leq C(x) + C(r \mid x) + 2 \log |x| + O(1)$ by using the standard way of appending two strings in a self-delimiting manner.

For the other inequality, we tweak the proof of the chain rule. We fix $x$ and $r$ and let $t = C(x,r)- |r| + \log (1/\epsilon)$.  By the standard counting argument, with probability $1 - \epsilon$, $t$ is positive, which we assume from now on. 
Since $C(x,r) \leq  |x| + |r| + 2 \log |x| + O(1)$ and $C(x,r) =  t + |r| - \log (1/\epsilon)$, we get $t \leq |x| + 2 \log |x| + \log(1/\epsilon) + O(1)$, and from here $\log t \leq O (\log (|x|/\epsilon))$, a fact which we will use later.

Now, let $A = \{(x',r') \mid C(x',r') \leq t + |r|\}$. For each $x'$, let $A_{x'}=\{r' \mid (x',r') \in A\}$. Let $e$ be such  that $2^{e-1} < |A_x| \leq 2^e$. Note that by taking into account that $|r| = f(|x|)$,
\begin{equation}
\label{e:eqqq1}
C(r \mid x) \leq e + 2 \log t + O(1).
\end{equation}
Next, let $H = \{x' \mid A_{x'} > 2^{e-1}\}$. We have $|H| \leq |A|/2^{e-1} \leq 2^{t+|r| + 1}/ 2^{e-1} = 2^{t+|r| - e +2}$.
Note that $x$ can be described by its rank in an enumeration of $H$, and for the enumeration of $H$ we need $t$ and $e$. By writing the rank of $x$ on exactly $t+|r|-e+2$ bits, and keeping in mind that $|r|$ can be computed from $|x|$, we obtain
\begin{equation}
\label{e:eqqq2}
C(x) \leq t+|r|-e+2 + 2 \log t + 2 \log |x| + O(1).
\end{equation}
By adding equations~\eqref{e:eqqq1} and~\eqref{e:eqqq2}, we obtain $C(x) + C(r \mid x) \leq t + |r| + 4 \log t + 2 \log |x| + O(1)$. 
Therefore,
\[
\begin{array}{rl}
C(x,r)  & = t + |r| - \log(1/\epsilon) \\
& \geq C(x) + C(r \mid x) - 4 \log t - 2\log |x| - \log(1/\epsilon) - O(1) \\
&\geq C(x) + C( r \mid x) - O(\log(|x|/\epsilon)).
\end{array}
\]
\end{proof}
\begin{lemma}
\label{l:comm-inf}
For all $\xa,\xb,z$,
$C(z) \le^+ C(z \mid \xa) + C(z \mid \xb) + I(\xa:\xb).$
\end{lemma}
\begin{proof}
For all $\xa,\xb,z$ we have  
$$
C(\xa,\xb \mid z) \le^+ C(\xa \mid z) + C(\xb \mid z).
$$
We add $2C(z)$ to both sides of this inequality and get 
$$
C(z) + C(\xa,\xb,z) \le^+ C(\xa,z) + C(\xb,z).
$$
Now we make it slightly weaker,
$$
C(z) + C(\xa,\xb) \le^+ C(\xa,z) + C(\xb,z).
$$
It follows
$$
C(z)   \le^+ C(z \mid \xa)+ C(z \mid \xb)+ C(\xa) + C(\xb)- C(\xa,\xb),
$$
which is equivalent to  \eqref{comm-inf}, ending the proof.
\end{proof}

\begin{lemma}\label{l:calculations}
(a) For all strings $w,x,y$ of length $O(n)$ 
\[
 C(w \mid x)  +  C(w\mid y)  + I(x:y ) - I(x:y \mid w) - C(w\mid x,y) \eqp  C(w) 
\]
(b)  For all strings $x,y,z,t$  of length $O(n)$, if $C(z\mid x,t)\eqp  C(z\mid y,t)\eqp 0$, then
\[
I(x:y \mid z,t) \eqp  I(x:y \mid t) - C(z\mid t)
\]
\end{lemma}
\begin{proof}
 (a) A straightforward calculation gives
  \[
 \begin{array}{lcl}
    C(w \mid x)  +  C(w\mid y)  + I(x:y ) - I(x:y \mid w) - C(w\mid x,y)  \\
   {} \eqp  C(x,w) - C(x) + C(y,w) - C(y)  + C(x) + C(y) - C(x,y)  \\
   \ \ \ \ {}  - C(x,w) - C(y,w) + C(x,y,w) + C(w)  - C(x,y,w) + C(x,y) =
    C(w)
   \end{array}
   \]   
(b) We have
 \[
 \begin{array}{ll}
 I(x:y \mid z,t)  \eqp  C(x,z,t) + C(y,z,t) - C(x,y,z,t)- C(z,t) &\\
 \ \  {} \eqp  C(x,t)+C(y,t)  -  C(x,y,t)-  C(z, t)  &\text{since  $C(z\mid x,t)\eqp  C(z\mid y,t)\eqp 0$}\\
  \ \ {} \eqp  C(x,t)+C(y,t)  -  C(x,y,t)-   C(t) - C(z \mid  t) &\text{the chain rule}\\
  \ \ {} \eqp  I(x:y\mid t) - C(z \mid  t) 
    \end{array}
   \]   
\end{proof}

\subsection{Proof of Lemma~\ref{lemma-eq-j}} %
For all $\xa,\xb,\xc,z$ we have
 $$
 C(\xa,\xb,\xc\mid z) \le^+ C(\xa\mid z)+ C(\xb \mid z) + C(\xc\mid z),
 $$
which rewrites to 
 $$
 2C(z)+ C(\xa,\xb,\xc, z) \le^+ C(\xa,z)+ C(\xb,z) + C(\xc,z).
 $$
Now we use the assumption that $z$ is a deterministic function of $\xa$, of $\xb$, and of $\xc$ :
 $$
 2C(z)  \le^+ C(\xa)+ C(\xb) + C(\xc)- C(\xa,\xb,\xc),
 $$
and the lemma is proven.

\section{Final comments}
We make several considerations regarding secret key agreement protocols. They are not required for an understanding of the main results. The  next subsections can be read independently.

\label{s:final}
\subsection{Time-efficient secret key agreement protocols.}
\label{s:time-efficient}
The secret key agreement protocol in the proof of Theorem~\ref{t:lower} is computable but highly non-efficient. The only slow stage of the protocol is Step~2, where Bob reconstructs $x$ given his input string $y$ and the fingerprint $\pxy$ obtained from Alice. At this stage, Bob has to simulate all programs of size $C(x\mid y)$ until he obtains a string matching the fingerprint $\pxy$.  All the other stages of the protocol can be implemented in polynomial time (to this end we need to use an effective version of an extractor in the definition of fingerprints; this increases the overhead in communication complexity from $O(\log n)$ to $\mathrm{poly}(\log n)$, which is still negligible comparative to the size of $\pxy$; for details see~\cite{zim:c:kolmslepianwolf}).

We cannot make Bob's computation effective in general, but we can do it for some specific pairs of inputs $(x,y)$. Actually, we can make the entire communication protocol fast, if there is a way to communicate $x$ from Alice to Bob so that
 \begin{itemize}
 \item communication complexity of this stage (and therefore the information revealed to the adversary) remains about  $C(x\mid y)$,
 \item all computations are performed by Alice and Bob in time $\mathrm{poly}(n)$.
 \end{itemize}
\emph{Example 1.} (Discussed in Introduction, p.~\pageref{line-point-example}.)
Let Alice get a random line  $x$ in the affine plane over the finite field with $2^n$ elements, and Bob get a random point $y$ on this line.    For \emph{most} inputs of this type we have
  $
 C(x\mid y) = n\pm O(\log n),
 $
 and there exists a simple way to transfer $x$ from Alice to Bob with communication complexity $n$ (Alice just sends to Bob the slope of her affine line, and Bob draws a line with this slope incident to his point). Thus, we see once again that for this simple example there exists an effective (polynomial-time) communication protocol to agree on  common secret key of size $\approx n$ bits.

\emph{Example 2.}
Let Alice and Bob get $n$-bits strings $x$ and $y$ respectively, and 
the Hamming distance between these strings is at most $\delta n$ for some constant $\delta<1/2$. 
For \emph{most} inputs of this type we have
  \[
 C(x\mid y) = h(\delta )n\pm O(\log n),\ I(x:y) = (1-h(\delta))n\pm O(\log n),
 \]
where $h(\delta) = \delta \log\frac1\delta+(1-\delta)\log\frac1{1-\delta}$. 
Can we transfer  $x$  from Alice to Bob with communication complexity  $h(\delta )n+o(n)$? It turns out that such a protocol exists; moreover, there exists a communication protocol with asymptotically optimal communication complexity and polynomial time computations, see \cite{smith2007scrambling,guruswami2010codes,guruswami2016optimal}. Plugging this protocol in our proof of Theorem~\ref{t:upper} we conclude that 
on most pairs of inputs $(x,y)$ of this type Alice and Bob can agree on a common secret key of size $(1-h(\delta))n- o(n)$, with poly-time computations for both parties.

\subsection{On pseudo-deterministic protocols for key agreement}
The communication protocols that we have analyzed (defined in Section~\ref{s:defprotocol})  are randomized. In particular, this means that,  for the same pair of inputs $\xa,\xb$,  Alice and Bob can  end up with different values of the common secret key $z$, since the result depends on the randomness used in the protocol. It seems that for all practical reasons, running  the same communication protocol twice on the same pair of inputs (but with different realizations of random bits) is useless and even harmful: in two independent conversations,  Alice and Bob would  divulge too  much information about $x$ and $y$. However, there remains a natural theoretical question: can we adjust the communication protocol so that for most values of random bits, Alice and Bob agree on one and the same value of the common random key $z$? Such a protocol could be called \emph{pseudo-deterministic} or \emph{Bellagio} protocol, similarly to the  pseudo-deterministic algorithms introduced  in \cite{goldwasser2012pseudo}.

In what follows we show that in some sense the answer to this question is positive. More precisely,  we can  prove a version of  Theorem~\ref{t:lower} with only slightly weaker property of secrecy (we divulge $O(\log n)$ bits of information regarding the resulting secret key) with a protocol where Alice and Bob with high probability agree on  one ``canonical'' value of $z$.

\begin{theorem}[Lower bound  with a unique key]
\label{t:lower-unique}
There exists a secret key agreement protocol with
the following property: For every $n$-bit strings $x$ and $y$, for every constant $\epsilon > 0$,  if Alice's input $\xa$ consists of $x$,  the complexity profile
of $(x, y)$ and $\epsilon$, and Bob's input $\xb$  consists of $y$, the complexity profile of $(x, y)$ and $\epsilon$, then, with probability
$1 - \epsilon$,  the shared secret key is a string $z$ such that,
$C(z \mid t) \geq  |z| - O(\log( n/\epsilon))$ and $|z| \geq I(x : y) - O(\log( n/\epsilon))$,
where $t$ is the transcript of the protocol. 

Moreover, 
with probability $1 - \epsilon$,   Alice and Bob agree on one ``canonical''  value of $z=z(x,y)$.
\end{theorem}

\begin{proof}[Sketch of the proof.]
Let us start with the simplified  scheme ``in a nutshell'' explained on p.~\pageref{communication-protocol--nutshell}.
In this protocol,  we use two randomized ``digital fingerprints'' of $x$ denoted $h_1(x)$ (which is sent by Alice to Bob)
and $h_2(x)$ (which is taken as the common secret key). In the argument on p.~\pageref{communication-protocol--nutshell} both these digital fingerprints were computed as random linear mappings of $x$.  Observe that the scheme works fine if the second digital fingerprint  is computed in a different way, without randomization. For example, we may assume that Alice and Bob search for  the first stopping program of length $C(x)$ that prints $x$, and then take the first $C(x) - C(x\mid y)$ bits of this program as the secret key. 

Clearly, in every execution of the new protocol, Alice and Bob agree on one and the same value of $z$ assuming that, given $y$ and $h_1(x)$ (and the complexity profile of $(x,y)$),  Bob reconstructs correctly Alice's input $x$, which is true with probability close to $1$. Moreover, it can be shown that with high probability the value of $h_1(x)$ contains only $ O(\log (n/\epsilon))$ bits of information regarding this ``canonical'' value of $z$, so the secrecy condition is preserved.

The same idea can be used to adjust the  proof of Theorem~\ref{t:lower}, where the value of $h_1$ is computed in a subtler way (with  randomness extractors  involved in the proof of Theorem~\ref{t:compression}). Roughly speaking, the random ``fingerprint'' of $x$ computed in the proof of Theorem~\ref{t:compression}  with high probability has only negligibly small mutual information with the fixed in advance ``canonical'' value of $z$.
We omit the details.
 \end{proof}
 
 We stress that the  ``secret key uniqueness''  achieved in Theorem~\ref{t:lower-unique} is a property of a specific protocol, and this property  does not follow automatically from the definition of shared secret key agreement.
 Moreover, we can construct a communication protocol so that on each new application of the protocol, Alice and Bob agree on a random value of the secret key $z$ that is \emph{uniformly distributed} on the set of $\{0,1\}^{k}$ with $k=I(x:y)$, and, therefore, the instances of the key $z$ obtained in different executions of the protocol are with high probability independent with each other.
(The protocol still has a positive  error probability, which means that with a small probability $\epsilon>0$ Alice and Bob do not agree on a common key and accomplish the communication with different values of $z$.)
To achieve this property, we should modify the protocol from Theorem~\ref{t:lower} as follows.
First of all,  Alice and Bob should agree on some common key $z$ using the old version of the protocol (and this $z$ may be distributed  non-uniformly). Then,  Alice choses a random bit string $w$ of length $|z|$ and openly sends it to Bob.  At last,  Alice and Bob XOR the bits of $z$ with the bits of $w$. The resulting string $z'$ is a uniformly distributed common secret key.

\smallskip

In the arguments above we used specially designed protocols.
In the ``standard'' communication protocol defined in the proof of Theorem~\ref{t:lower} 
the key $z$ is not uniquely defined by $x,y$, and it is not uniformly distributed either.
The specific attributes of the shared key produced by  a  protocol are determined by the
details of the construction and by the ratio between $C(x)$ and $I(x:y)$, and are positioned between the two extreme cases that we have presented above.
A more detailed analysis of generic communication protocols (e.g., a more precise estimation of $I(z_1:z_2)$ for two keys $z_1$ and $z_2$ obtained in two independent realizations of the protocol from the proof of Theorem~\ref{t:lower})  requires subtler considerations, and we do not discuss this question here.

\subsection{Secret key agreement with an additional  private communication channel}\label{s:private-channel}

In this section, we discuss a slightly more general model of communication that combines a public and private channels.
We assume that the \emph{private} communication channel is not visible to the adversary. There is a naive and straightforward way to gain from such a channel: Alice tosses her private random coin,  sends the obtained bits to Bob via this private channel, and afterwards  Alice and Bob can use these bits as a common secret key.  This naive idea is actually the optimal usage of the private channel. There is no better  way to use the secret communication: if Alice and Bob are given inputs $\xa$ and $\xb$, and they send to each other $s$ bits of information via the private communication, they cannot agree (with probability $1-\epsilon$)  on a common secret key of complexity greater than   $I(\xa:\xb)+s+O(\log (n/\epsilon))$. Indeed, denote by $t_p$ the transcript of the communications via the private channel. Then the common secret $z$ can be computed from $(\xa,\ra,t_p)$ as well as from $(\xb,\ra,t_p)$. It follows that
\[
 \begin{array}{rcl}
 C(z) &\leqp& C(z\mid \xa,\ra,t_p)  +  C(z\mid \xb,\rb,t_p) + I(\xa,\ra,t_p : \xb,\rb,t_p)      \\
         &\leqp& I(\xa,\ra,t_p : \xb,\rb,t_p)   \leqp  I(\xa,\ra: \xb,\rb) + C(t_p) \\ 
         &\le& I(\xa : \xb) + C(t_p) +  O(\log (n/\epsilon))  \\
          &\le& I(\xa : \xb) + s +  O(\log (n/\epsilon))          
 \end{array}
\]
(the third inequality holds with probability $1-\epsilon$).

Let us consider now  a model where Alice and Bob communicate via two channels: via an  `expensive'  private channel of capacity $s$ bits    and a `cheap' public channel of unbounded capacity (accessible to the adversary). 
As above, let $\xa$ and $\xb$ be input strings of length $n$  on which the  protocol succeeds with error probability $\epsilon$ and randomness deficiency $\delta(n) = O(\log n)$, and let $z$ be the random string that
is the shared secret key output by the protocol, \ie, a string satisfying relations (\ref{e:e1}) and (\ref{e:e2}). 
We combine the simple observation above with the proof of Theorem~\ref{t:upper} and conclude  with the following result:
with probability at least $1-O(\epsilon)$, if $n$ is sufficiently large, $|z| \leq I(\xa : \xb)+s+O(\log (n/\epsilon))$.

The `light upper bound' from Section~\ref{s:warmup} does not apply to the model with an additional private channel. Actually, the traffic of size $s$ over the private channel in some sense can be counted as an ``overhead'' in 
inequalities \eqref{I(a:b:t)-is-positive} and \eqref{I(a:b:t|c)-is-positive}.
More precisely, we can prove a version of  Lemma~\ref{lemma-I(a:b:t)-is-positive} with an extra term corresponding to the private channel.
If a message  sent by Alice is computed as a function of its own input $\xa$ and of the data $t_\mathrm{priv}$ sent via the private  channel, than
similarly to \eqref{I(a:b:t)-is-positive}  we can prove that
  \begin{equation}
   I(\xa:\xb \mid f(\xa,t_\mathrm{priv}))  \le^{+}   I(\xa:\xb ) + C(t_\mathrm{priv}).\nonumber
  \end{equation}
Also, if a message sent by Alice is computed as a function of her own input $\xa$, the publicly visible ``prehistory'' of the protocol $t_{\mathrm{pub}}$
(i.e., the bits sent earlier via the public channel), and of the secret part of the transcript  $t_\mathrm{priv}$ (sent via the private channel), 
then similarly to \eqref{I(a:b:t|c)-is-positive} we can prove that
  \begin{equation}
   I(\xa:\xb \mid  g(\xa,t_\mathrm{pub} ,t_\mathrm{priv}), t_\mathrm{pub} )  \le^{+}   I(\xa:\xb \mid t_\mathrm{pub}) + C(t_\mathrm{priv}).\nonumber
     \end{equation}
Using these inequalities and keeping in mind that the size of the private part of the communication $C(t_\mathrm{priv})$ is bounded by $s$, 
we  can adapt the argument from p.~\pageref{e:e3} with a suitable correction:
in the chain of inequalities \eqref{e:e3}  we need to count $s$ bits of the ``overhead'' for each round of communication via the public channel. 
For a $k$-round protocol this would give the bound $|z| \leq I(\xa : \xb)+ks+O(\log (n/\epsilon))$, which is too weak for non-constant $k$.

We revisit these observations in Section~\ref{s:discussion}.

\subsection{Secret key agreement: the Shannon framework  $\mbox{vs.}$ the  Kolmogorov framework}\label{s:discussion}

\textbf{Upper and lower bounds on the size of the shared secret key.}
\label{subsection-shannon-vs-kolmogorov}
In the Introduction we discussed the general connections between IT and AIT.
In what follows we discuss some more precise and formal relations between Shannon-type and Kolmogorov-type 
theorems on secret key agreement.

Actually the upper and lower bounds on the size of the common key in the framework of Kolmogorov complexity 
(Theorem~\ref{t:lower} and Theorem~\ref{t:upper}) formally imply similar bounds
in Shannon's framework  for i.i.d. pairs variables and even for stationary ergodic sources (Theorem~1 in \cite{tyagi2013common}, see also \cite{leu:t:secret,ben-bra-rob:j:privacyamplification,mau:j:seckey,ahl-csi:j:seckeyone}). Indeed, it is well known
that  a sequence of i.i.d. random variables (or, more generally, the outcome of a stationary ergodic source)
gives with a high probability a string of letters whose  Kolmogorov complexity is close to Shannon's entropy of the random source
(this statement was announced by L.~Levin in \cite[Proposition~5.1]{zvo-lev:j:kol} and published with proof in \cite{brudno1982entropy};
see also  \cite{horibe2003note}). 
Hence, if Alice gets a value of
 $(X_1,\ldots, X_n)$,  Bob gets a value of $(Y_1,\ldots, Y_n)$, and the sequence of pairs $(X_i,Y_i)$ is a stationary ergodic source, then with high probability the mutual information  (in the sense of Kolmogorov complexity) between Alice's and Bob's individual inputs  is close to 
the mutual information between these two random sequences (in the sense of Shannon's entropy). 

Thus, to prove the positive  statement in Shannon's framework (Alice and Bob can agree on a secret key of size about $ I(X_1,\ldots, X_n : Y_1,\ldots, Y_n)$) we can apply  the communication protocol from the proof of Theorem~\ref{t:lower}.  Notice that in Theorem~\ref{t:lower} Alice and Bob need to know the complexity profile of the input data; we cannot know the \emph{exact} values of  Kolmogorov complexity for a pair of randomly chosen random strings, but it is enough to provide Alice and Bob with the known \emph{expected} value of the complexity profile of their input data (see Remark~\ref{r:lower-approx} on p.~\pageref{r:lower-approx}). 
In the case when the complexity profile of the values of  $X_1,\ldots, X_n$ and $Y_1,\ldots, Y_n$ is `typical' (close to the average values), Theorem~\ref{t:lower} guarantees that the resulting common key $z$ has a large enough Kolmogorov complexity conditional on the transcript of the protocol.  Otherwise, (in case when the complexity profile of the input data is not typical) the result of the protocol can be meaningless; however, this happens with only a small probability.
So,  assuming that Alice and Bob are given random values of  $X_1,\ldots, X_n$ and $Y_1,\ldots, Y_n$ respectively, we conclude that the outcome of the protocol (the common secret key) is distributed mostly on strings of high Kolmogorov complexity. It is not hard to show that such a distribution must have  Shannon's entropy (conditional on the value of the transcript) close to $ I(X_1,\ldots, X_n : Y_1,\ldots, Y_n)$.

In a similar way, our negative result (the secret common key cannot be made bigger, Theorem~\ref{t:upper}) implies a similar result in Shannon's setting. Indeed, assume that there exists a communication protocol that permits to Alice and Bob to agree on a common secret  on random inputs 
$X_1,\ldots, X_n$ and $Y_1,\ldots, Y_n$ with a probability close to $1$. If the pair of random sources is 
stationary and ergodic, then  with high probability the pair of input values has a  complexity profile close to the corresponding values of Shannon's entropy. It follows from Theorem~\ref{t:upper} that the Kolmogorov complexity of the secret key is not much bigger than the mutual information between $X_1,\ldots, X_n$ and $Y_1,\ldots, Y_n$. It remains to notice that the Shannon's entropy of the common secret key $Z$ (as a  random variable) cannot be much bigger than the average  Kolmogorov complexity of its value (as a binary string).

In Shannon's framework,   we can consider a hybrid communication model (similar to those in  Section~\ref{s:private-channel}),  where Alice and Bob can communicate via a private channel of bounded capacity  (hidden from the adversary)  and a public channel of unbounded capacity (accessible to the adversary). 
The connection between Shannon's and Kolmogorov's settings discussed above combined with the observation made in Section~\ref{s:private-channel} implies  that,
 in this model,  Alice and Bob cannot  agree on a common secret key of entropy greater than
 \[
    I(X_1,\ldots, X_n : Y_1,\ldots, Y_n) +[\text{capacity of the private channel}].
 \]

\textbf{Communication complexity of the protocol.}
  Another issue  that can be studied in both settings (for Shannon's and Kolmogorov's approaches to the key agreement problem) is the optimal communication complexity
 of the protocol. 
 Tyagi investigated in \cite{tyagi2013common} this problem for the Shannon's framework (for i.i.d. sources)  and showed that in some cases (for some distributions
 $(X_i,Y_i)$) Alice and Bob
 need to send to each other at least  
  \[
    \min(
    H(X_1,\ldots, X_n \mid Y_1,\ldots, Y_n),
    H(Y_1,\ldots, Y_n \mid X_1,\ldots, X_n)
    )
  \]
 bits of information to agree on a common secret key of size $ I(X_1,\ldots, X_n : Y_1,\ldots, Y_n)$. Moreover,  Tyagi found a complete description 
 (a so-called \emph{one letter characterization}) of the optimal communication complexity for an arbitrary distribution on Alice's and Bob's inputs. 
 This result is similar to our Theorem~\ref{thm:comm-complexity}, which claims that Alice and Bob
 need to send at least 
   $
    \min(
    C(x\mid y),
    C(y\mid x)
    )
  $
  bits of information to agree even on a common secret key of size $ \delta I(x : y)$.
  However, these results are  incomparable with each other (one does not follow from another):
  \begin{itemize}
  \item Tyagi investigated protocols that agree on a common secret of maximal size  $I(X_1,\ldots, X_n : Y_1,\ldots, Y_n)$, 
  while Theorem~\ref{thm:comm-complexity} applies even if Alice and Bob agree on a secret of size $\delta I(x:y)$ for a small $\delta>0$.
  \item  \cite{tyagi2013common} proposed  a one letter characterization of the communication complexity in terms of the distribution $(X_i,Y_i)$, 
  which arguably cannot be reformulated in the setting of Kolmogorov complexity.
  \item The result in  \cite{tyagi2013common}  applies to communication protocols with private randomness, while Theorem~\ref{thm:comm-complexity} is proven only for protocols with public randomness.
  \end{itemize}
  
\textbf{One-shot paradigm.}
Most known results on the protocols for secret key agreement in Shannon's framework are proven with the assumption that the input data available to Alice and Bob are i.i.d. random variables or at least stationary ergodic random sources. We mentioned above that for this class of inputs (stationary ergodic random sources) the bounds on the optimal size of the secret in Shannon's framework   can be deduced from analogous results on Kolmogorov complexity. But, actually Theorem~\ref{t:lower} and Theorem~\ref{t:upper} apply in more general  settings. We can prove similar bounds for random inputs obtained \emph{in one shot}, without the property of ergodicity.

In many natural instances  of  the secret key agreement problem,  the input data are far from being ergodic, so the classic technique does not apply. In Section~\ref{s:time-efficient},  we  discussed two  examples of this kind: in the first one,  Alice is  given a line in the affine plane, and Bob is given a point in this line; in the second one,  Alice and Bob get $n$-bits strings $x$ and $y$ respectively,  and the Hamming distance between these strings is at most $\delta n$ for some constant $\delta<1/2$.   These examples can be naturally reformulated in the probabilistic setting: we can introduce the uniform distribution on the set of all valid pairs of inputs. For the probabilistic versions of these examples, the matching upper and lower bounds on the size of the common secret key can be easily deduced from Theorem~\ref{t:lower} and Theorem~\ref{t:upper}.

To conclude, the standard results on the secret key agreement  deal with the paradigm  where \emph{the protocol works properly for {most} randomly chosen inputs} (which is typical for  information theory),  while in our approach we prove a somewhat stronger statement :  \emph{for {each} valid pair of input data the protocol works properly with high probability} (which is typical for the theory of communication complexity).



\subsection{Communication protocols with logarithmic advice}
\label{s:protocols-with-advice}

In this section we show that in some setting the `light upper bound' in Section~3 is more robust and flexible  than the more technical proof of Theorem~4.2.
 
Despite many prominent parallels between Shannon's and Kolmogorov's variants of information theory, there is no  canonical way to translate the basic properties from one framework to another. The rule of thumb is to substitute Kolmogorov complexity instead of Shannon's entropy and to change all constraints of type `\emph{a random variable $X$ is a function of a random variable $Y$}' in Shannon's case to `\emph{a string $x$ has small Kolmogorov complexity conditional on  $y$}' in Kolmogorov's version. For example, the classic Slepian--Wolf theorem essentially claims  that, for jointly distributed $X,Y$ (that technically must be a sequence of i.i.d. pairs), there exists another random variable $Z$ such that 
  \[
  \left\{
  \begin{array}{l}
  Z\text{ is a deterministic function of } X,\\
  \text{number of bits in }Z \approx H(X\mid Y),\\
  X\text{ can be recovered from }Y\text{ and }Z \text{ with high probability}.
  \end{array}
  \right.
  \]
The counterpart of this result for Kolmogorov complexity is known as Muchnik's theorem, \cite{ muc:j:commoninfo}: for all strings $x,y$ of length $n$ there exists a string $z$ such that
 \[
  \left\{
  \begin{array}{l}
  C(z\mid x)=O(\log n),\\
  |z| = C(x\mid y)+O(\log n),\\
  C(x\mid y,z) = O(\log n).
  \end{array}
  \right.
  \]
We can use a similar logic to translate in the language of Kolmogorov complexity the general notion of an interactive communication protocol. A protocol adapted to the framework of Kolmogorov complexity can be viewed as a protocol where on each step Alice and Bob get  logarithmic `advice'  strings from a trustable Merlin.

In what follows we give a more formal definition of a protocol with advice bits.
 Let $k$ be a positive integer. We say that Alice and Bob given inputs $\xa$ an $\xb$ compute a common value $z$ in  a $k$-rounds \emph{protocol with $q$ bits of advice on each round},  given inputs $\xa$ and $\xb$, if there exists  is a sequence of strings (`messages') $x_1,y_1,x_2,y_2,\ldots,x_k, y_k$ such that 
 \[
\begin{array}{l} 
C(x_1 \mid \xa) \le q,\\
C(y_1 \mid \xb, x_1) \le q,\\
C(x_2 \mid \xa,  x_1, y_1) \le q, \\
C(y_2 \mid \xb,  x_1, y_1, x_2) \le q,\\

\vdots \\

C(x_k \mid  \xa, x_1,y_1, \ldots, x_{k-1}, y_{k-1}) \le q,\\
C(y_k \mid  \xb, x_1,y_1, \ldots, x_{k-1}, y_{k-1})\le q,\\
C(z\mid  \xa, x_1,y_1, \ldots, x_{k}, y_{k}) \le q, \\
C(z\mid  \xb, x_1,y_1, \ldots, x_{k}, y_{k}) \le q.
\end{array}
\]
In a similar way, we can define the  \emph{randomized} protocols with advice, where Alice and Bob are given in addition random $\ra$ and $\rb$. 
For $c>0, \epsilon >0$,  we say that  Alice and Bob   $(\epsilon, c)$-\emph{succeed} on an  input pair $(\xa, \xb)$  to agree on a  common secret key in a $k$-rounds protocol with $q$ bits of advice, if 
with probability $(1-\epsilon)$ over $\ra, \rb$,  there exists a string $z$  
such that  Alice and Bob can compute common value $z$ in a $k$-rounds protocol with advice as defined above, and 
\begin{equation}\nonumber
C(z \mid t) \geqp |z| - c \log n,
\end{equation}
where $t:=(x_1, y_1, \ldots, x_k, y_k)$ is a transcript of this protocol.

\begin{remark}
In a conventional definition of a communication protocol, every next message  sent by Alice or Bob is computed as a  function of all previous  messages and of the input given to Alice or Bob respectively  (this function is one and the same for all pairs of inputs). 
Thus, every next message is ``simple'' conditional on the previous messages and one party's input in a very strong sense.
Such a  schema is unusual in the context of algorithmic information theory, where the relation ``$a$ is simple conditional on $b$'' is typically understood as $C(a\mid b) \sim \log (C(a)+C(b))$.   This observation motivates the given above definition of a communication protocol with advice.  As a matter of fact, we substitute  the standard definition of a communication protocol by a chain of conditions   ``every new message is simple conditional on the prehistory of the communication and one party's input.''
Considering this motivation, we believe that a natural measure if ``simplicity'' is logarithmic (i.e., on each round of the protocol we can use $O(\log n)$
bits of advice, where $n$ is the length of the inputs).

Formally speaking, we might consider another communication model where the advice is provided  only for some selected rounds, while on other steps of the protocol the messages are computed in a conventional way, by a fixed in advance deterministic function. In this case we should bound \emph{the total} number of advice bits distributed  over all rounds. However, such a mixed model lacks a natural motivation.
\end{remark}

Any conventional communication protocol (without advice) can serve as a protocol with advice strings of size $q=O(1)$, so the lower bound from Theorem~\ref{t:lower} remains true in the new setting. Also the proof of the `light' upper bound form Section~\ref{s:warmup}  applies to the protocols with advice strings  of size $q=O(\log n)$, and we conclude that in the new setting Alice and Bob still cannot agree on a common secret key of size greater than $I(\xa:\xb)$. 
Notice that the proof of Theorem~\ref{t:upper} does not apply to protocol with advice, so we cannot extend this upper bound to non-constant number of protocols. But this is not a real loss: the protocols with advice probably make no sense for non-constant number of rounds (since too much information can be  hidden in the bits of advice).

 \subsection{Weak version of negative results for multi-party protocols}\label{s:weak-upper-bound-multivar}

In this section we  prove a weak version of the upper bound on the length of the longest secret key that three parties can agree upon. The proposed argument follows the ideas from the proof of a similar result in \cite{csi-nar:j:seckey} in Shannon's setting (it is based on classic information inequalities).
While this technique works pretty well in IT, in our case it covers only the protocol with $O(1)$ rounds and with  the number of random bits  bounded polynomially in the input length. Thus, the result proven in this section is much weaker than Theorem~\ref{t:upperbdthree}. However, it is instructive to see how far we can go in the setting of AIT with the classic technique  borrowed from IT.

\begin{theorem}
\label{t:upperbd} Let us consider a  $3$-party  protocol  for secret key agreement with error probability $\epsilon$, having a constant number of rounds, and where the number of random bits is bounded polynomially in the input length.
Let $(\xa, \xb, \xc)$ be a $3$-tuple of  $n$-bit strings on which the protocol succeeds. Let $z$ be the random variable which represents the secret key computed from the input $(\xa, \xb, \xc)$ and let $t$ be the transcript of the protocol that produces $z$. Then,  for sufficiently large $n$,  with probability $1-O(\epsilon)$, 
\[
C(z \mid t) \leq C(\xa,\xb,\xc) - \co(\xa, \xb, \xc)+ O(\log (n/\epsilon)),
\]
where the constants in the $O(\cdot)$ notation depend on the universal machine,  the number of rounds and the protocol, but not on $(\xa, \xb, \xc)$.
\end{theorem}
\begin{proof}
We fix a triplet of $n$-bit input strings $(\xa, \xb, \xc)$ and we consider  some randomness $r = (\ra, \rb, \rc)$, on which the protocol produces a string $z$ which satisfies the relations~(\ref{e:e4}) and~(\ref{e:e5}). Recall that the set of such $r$'s has probability $1-\epsilon$. To simplify the notation, we merge the input of each participant with the corresponding random strong and denote
\[
 \xa' :=\langle \xa, \ra\rangle,\   \xb' :=\langle \xb, \rb\rangle,\   \xc' :=\langle \xc, \rc\rangle.
\]
Observe that with high probability the random strings $\ra, \rb, \rc$ are independent of $\xa, \xb, \xc$, and, therefore, 
\[
C(\xa',\xb',\xc')  \eqp C(\xa,\xb,\xc) + |\ra| + |\rb| + |\rc|
\]
and
\[
\co(\xa',\xb',\xc') \eqp \co(\xa, \xb, \xc) + |\ra| + |\rb| + |\rc|.
\]
Thus, to prove the theorem, it is enough to show that
\[
C(z \mid t) \leq C(\xa',\xb',\xc') - \co(\xa', \xb', \xc')+ O(\log (n/\epsilon)).
\]
To this end  we  transform the value of $C(\xa', \xb', \xc')$ as follows:
\begin{equation}
\label{e:eq1}
\begin{array}{rl}
C(\xa',\xb',\xc' ) \eqp & C(t,z,\xa',\xb',\xc') \\
\eqp & C(t_1 ) + C(t_2 \mid t_1, ) + C(t_3 \mid t[1:2]) + \ldots + C(t_{3k} \mid t[1:3k-1]) \\
+  & C(z\mid t) + C(\xa' \mid z, t) + C(\xb' \mid \xa',z,t) + C(\xc' \mid \xa', \xb', z, t).
\end{array}
\end{equation}
The first line follows because $t$ and $z$ are computed from $\xa, \xb, \xc$ and $r$, and the second line uses the chain rule.
Next, we break the  sum in the right hand side of equation~\eqref{e:eq1} into the sum $q_1(r) + q_2(r) + q_3(r)$, where the terms in the latter sum are defined as
\begin{equation}
\nonumber
\begin{array}{rcl}
q_1(r) & := &  C(t_1 ) + C(t_4 \mid t[1:3] ) + \ldots + C (t_{3k-2} \mid t[1:3k-3] ) + C(\xa' \mid z, t ) \\
\\

q_2(r) &: = & C(t_2 \mid t_1 ) + C(t_5 \mid t[1:4] ) + \ldots + C (t_{3k-1} \mid t[1:3k-2] )+ C(\xb' \mid \xa', z, t ), \\
\\

q_3(r) & := & C(t_3 \mid t[1:2] ) + C(t_6 \mid t[1:5] ) + \ldots + C (t_{3k} \mid t[1:3k-1] )+ C(\xc' \mid \xa', \xb', z, t ).
\end{array}
\end{equation}
By definition of $q_1(r), q_2(r), q_3(r)$ and \eqref{e:eq1} we have
\begin{equation}
\label{eq:abc-vs-q123}
C(\xa',\xb',\xc' ) \eqp C(z \mid t) + (q_1(r) + q_2(r) + q_3(r)).
\end{equation}
The following claim shows that $q_1(r), q_2(r)$ and $q_3(r)$ satisfy relations similar to those that define $S(\xa,\xb,\xc)$. 
\begin{claim}
\label{c:chain}
\hspace*{\fill}

\begin{itemize}
\item[(1)] $q_1(r) \geqp C(\xa' \mid \xb', \xc' )$, $q_2(r) \geqp C(\xb' \mid \xa', \xc' )$, $q_3(r) \geqp C(\xc' \mid \xa' , \xb' )$.

\item[(2)] $q_1(r) + q_2(r) \geqp C(\xa', \xb'  \mid \xc' )$,  $q_1(r) + q_3(r) \geqp C(\xa', \xc'  \mid \xb' )$,  $q_2(r) + q_3(r) \geqp C(\xb', \xc'  \mid \xa' )$.

\end{itemize}
\end{claim}
\begin{proof}[Proof  of Claim~\ref{c:chain}]
We show the first relation in~(1).
\begin{equation}
\label{e:eq2}
\begin{array}{rcl}
C(\xa' \mid \xb', \xc' ) &\eqp & C(t,z, \xa' \mid \xb', \xc' ) \\
&\eqp & C(t_1 \mid \xb', \xc' ) + C(t_2 \mid t_1, \xb', \xc' ) + C(t_3 \mid t[1:2], \xb', \xc' ) \\
&&{} +   C(t_4 \mid t[1:3], \xb', \xc' )  + C(t_5 \mid t[1:4], \xb', \xc' )  + C(t_6 \mid t[1:5], \xb', \xc' )  \\
 && \vdots \\
&&{}+   C(t_{3k-2} \mid t[1:3k-3], \xb', \xc' )  + C(t_{3k-1}\mid t[1:3k-2], \xb', \xc' )  \\
&& \hspace{6cm} {} + C(t_{3k} \mid t[1:3k-1], \xb', \xc' )  \\
 && {}+C(z \mid t,\xb', \xc' ) + C(\xa \mid t,z, \xb', \xc' ).
\end{array}
\end{equation}
The first line follows from the fact that $t$ and $z$ are computed from $\xa, \xb, \xc$ and $r$, and the second line follows from the chain rule. 

For every $\ell \in \{1, \ldots, k\}$,  $C(t_{3\ell-1} \mid t[1:3\ell-2], \xb', \xc' ) \eqp 0$, because 
$t_{3\ell-1}$ is computed from $t[1:3\ell-2], \xb', \xc' )$.
For the same reason, $C(t_{3\ell} \mid t[1:3\ell-1], \xb', \xc') \eqp 0$ and $C(z \mid t, \xb', \xc') \eqp 0$. Eliminating the terms that are $\eqp 0$ in equation~\eqref{e:eq2} and also the conditioning on $\xb'$ and $\xc'$ in all terms on the right hand side, we obtain 
\begin{equation}
\nonumber
\begin{array}{rl}
C(\xa' \mid \xb' , \xc') &\leqp C(t_1) + C(t_4 \mid t[1:3]) + \ldots + C (t_{3k-2} \mid t[1:3k-3]) + C(\xa' \mid z, t) \\
 & = q_1(r).
\end{array}
\end{equation}
The other relations in the claim  are proved in a similar way.
\end{proof}


Claim~\ref{c:chain} implies that $q_1(r) + q_2(r) + q_3(r) \geq \co(\xa', \xb', \xc')$. Combining this fact with relation~ \eqref{eq:abc-vs-q123},  
we obtain  $C(z \mid t) \leqp C(\xa',\xb',\xc') - \co(\xa', \xb', \xc')$, and we are done.
\end{proof}

 \subsection{Open problems}
\begin{openq}
{In Theorem~\ref{thm:comm-complexity} we establish a lower bound on how many bits Alice and Bob must communicate to agree on a common secret key. Our proof is valid only for communication protocols with public randomness. Is the same bound true for protocols with private sources of random bits?}
\end{openq}

\begin{openq}
All our communication protocols are randomized. It is natural to ask whether we can get rid of external randomness.
We conjecture that  for $(O(\log n), O(\log n))$-stochastic tuples of inputs the protocol can be made purely deterministic 
(though it would require very high computational  complexity), but this cannot be done in the general case.
The proof of this fact would require a better understanding of the nature of non-stochastic objects
(for a discussion of non-stochastic objects see \cite{shen1983concept} and \cite[Section~14.2]{suv:b:kolmenglish}).
\end{openq}

\section{Acknowledgments} We are grateful to Bruno Bauwens for his insightful comments. In particular we thank him for allowing us to reproduce his proof of the ``light lower bound" in Section~\ref{s:warmup} and for bringing to our attention the improved Kolmogorov complexity version of the Slepian-Wolf theorem~\ref{t:compressmulti}. We thank Tarik Kaced for attracting our attention to \cite{chan2015multivariate}.

\bibliography{theory-3}

\newcommand{\etalchar}[1]{$^{#1}$}
\begin{thebibliography}{CABE{\etalchar{+}}15}

\bibitem[AC93]{ahl-csi:j:seckeyone}
Rudolf Ahlswede and Imre Csisz{\'{a}}r.
\newblock Common randomness in information theory and cryptography - {I:}
  secret sharing.
\newblock {\em {IEEE} Trans. Information Theory}, 39(4):1121--1132, 1993.

\bibitem[ALPS10]{antunes-laplante}
Luis Antunes, Sophie Laplante, Alexandre Pinto, and Liliana Salvador.
\newblock Cryptographic security of individual instances.
\newblock {\em ICITS}, pages 195--210, 2010.

\bibitem[Bau18]{bau:t:kolmslepwolf}
Bruno Bauwens.
\newblock Optimal probabilistic polynomial time compression and the
  slepian-wolf theorem: tighter version and simple proofs.
\newblock {\em arXiv preprint arXiv:1802.00750}, 2018.

\bibitem[BBR88]{ben-bra-rob:j:privacyamplification}
Charles~H. Bennett, Gilles Brassard, and Jean-Marc Robert.
\newblock Privacy amplification by public discussion.
\newblock {\em SIAM Journal on Computing}, 17(2):210--229, 1988.

\bibitem[Bru82]{brudno1982entropy}
A.~A. Brudno.
\newblock Entropy and the complexity of the trajectories of a dynamic system.
\newblock {\em Trudy Moskovskogo Matematicheskogo Obshchestva}, 44:124--149,
  1982.

\bibitem[BZ14]{bau-zim:c:linlist}
Bruno Bauwens and Marius Zimand.
\newblock Linear list-approximation for short programs (or the power of a few
  random bits).
\newblock In {\em {IEEE} 29th Conference on Computational Complexity, {CCC}
  2014, Vancouver, BC, Canada, June 11-13, 2014}, pages 241--247. {IEEE}, 2014.

\bibitem[CABE{\etalchar{+}}15]{chan2015multivariate}
Chung Chan, Ali Al-Bashabsheh, Javad~B Ebrahimi, Tarik Kaced, and Tie Liu.
\newblock Multivariate mutual information inspired by secret-key agreement.
\newblock {\em Proceedings of the IEEE}, 3(10):1883--1913, 2015.

\bibitem[Cha66]{cha:j:length-of-programs}
Gregory~J. Chaitin.
\newblock On the length of programs for computing finite binary sequences.
\newblock {\em Journal of the ACM}, 13:547--569, 1966.

\bibitem[CK78]{csi-kor:j:secret}
Imre Csisz{\'{a}}r and J{\'{a}}nos K{\"{o}}rner.
\newblock Broadcast channels with confidential messages.
\newblock {\em {IEEE} Trans. Information Theory}, 24(3):339--348, 1978.

\bibitem[CMR{\etalchar{+}}02]{che-muc-rom-she-ver:j:commoninfo}
Alexey~V. Chernov, Andrei~A. Muchnik, Andrei~E. Romashchenko, Alexander Shen,
  and Nikolai~K. Vereshchagin.
\newblock Upper semi-lattice of binary strings with the relation "x is simple
  conditional to y".
\newblock {\em Theor. Comput. Sci.}, 271(1-2):69--95, 2002.

\bibitem[CN04]{csi-nar:j:seckey}
Imre Csisz{\'{a}}r and Prakash Narayan.
\newblock Secrecy capacities for multiple terminals.
\newblock {\em {IEEE} Trans. Information Theory}, 50(12):3047--3061, 2004.

\bibitem[GK73]{gac-kor:j:commoninfo}
Peter G{\'a}cs and J{\'a}nos K{\"o}rner.
\newblock Common information is far less than mutual information.
\newblock {\em Probl. Control Inf. Theory}, 2(2):149--162, 1973.

\bibitem[Gol12]{goldwasser2012pseudo}
Shafi Goldwasser.
\newblock Pseudo-deterministic algorithms.
\newblock In {\em STACS'12 (29th Symposium on Theoretical Aspects of Computer
  Science)}, volume~14, pages 29--29. LIPIcs, 2012.

\bibitem[GS10]{guruswami2010codes}
Venkatesan Guruswami and Adam Smith.
\newblock Codes for computationally simple channels: Explicit constructions
  with optimal rate.
\newblock {\em Foundations of Computer Science (FOCS), 2010 51st Annual IEEE
  Symposium on}, (723--732), 2010.

\bibitem[GS16]{guruswami2016optimal}
Venkatesan Guruswami and Adam Smith.
\newblock Optimal rate code constructions for computationally simple channels.
\newblock {\em Journal of the ACM (JACM)}, 63(4):35, 2016.

\bibitem[Hor03]{horibe2003note}
Yasuichi Horibe.
\newblock A note on {K}olmogorov complexity and entropy.
\newblock {\em Applied mathematics letters}, 16(7):1129--1130, 2003.

\bibitem[Kol65]{kol:j:kolmcomplexity}
Andrei~Nikolaevich Kolmogorov.
\newblock Three approaches to the quantitative definition of information.
\newblock {\em Problems Inform. Transmission}, 1(1):1--7, 1965.

\bibitem[KR13]{kaced2013conditional}
Tarik Kaced and Andrei Romashchenko.
\newblock Conditional information inequalities for entropic and almost entropic
  points.
\newblock {\em IEEE Transactions on Information Theory}, 59(11):7149--7167,
  2013.

\bibitem[KRV15]{kaced2015conditional}
Tarik Kaced, Andrei Romashchenko, and Nikolay Vereshchagin.
\newblock Conditional information inequalities and combinatorial applications.
\newblock {\em arXiv preprint arXiv:1501.04867}, 2015.

\bibitem[LYC76]{leu:t:secret}
Sik~Kow Leung-Yan-Cheong.
\newblock Multi-user and wiretap channels including feedback, July 1976.
\newblock Tech. Rep. No. 6603-2, Stanford Univ.

\bibitem[Mau92]{mau:j:cryptoprov}
Ueli~M. Maurer.
\newblock Conditionally-perfect secrecy and a provably-secure randomized
  cipher.
\newblock {\em Journal of Cryptology}, 5(1):53--66, 1992.

\bibitem[Mau93]{mau:j:seckey}
Ueli~M. Maurer.
\newblock Secret key agreement by public discussion from common information.
\newblock {\em {IEEE} Trans. Information Theory}, 39(3):733--742, 1993.

\bibitem[MR10]{muchnik-romash}
Andrei~A. Muchnik and Andrei~E. Romashchenko.
\newblock Stability of properties of {K}olmogorov complexity under
  relativization.
\newblock {\em Problems of information transmission}, 46(1):38--61, 2010.

\bibitem[MRS11]{mus-rom-she:j:muchnik}
D.~Musatov, A.~E. Romashchenko, and A.~Shen.
\newblock Variations on {M}uchnik's conditional complexity theorem.
\newblock {\em Theory Comput. Syst.}, 49(2):227--245, 2011.

\bibitem[Muc98]{muc:j:commoninfo}
Andrei~A. Muchnik.
\newblock On common information.
\newblock {\em Theor. Comput. Sci.}, 207:319--328, 1998.

\bibitem[Raz11]{razenshteyn-2011}
Ilya Razenshteyn.
\newblock Common information revisited.
\newblock {\em arXiv preprint arXiv:1104.3207}, 2011.

\bibitem[Rom00]{rom:j:mutualinfo}
Andrei Romashchenko.
\newblock Pairs of words with nonmaterializable mutual information.
\newblock {\em Problems of Information Transmission}, 36(1):3--20, 2000.

\bibitem[RRV02]{rareva:j:extractor}
Ran Raz, Omer Reingold, and Salil~P. Vadhan.
\newblock Extracting all the randomness and reducing the error in {T}revisan's
  extractors.
\newblock {\em J. Comput. Syst. Sci.}, 65(1):97--128, 2002.

\bibitem[Sha02]{shaltiel2002recent}
Ronen Shaltiel.
\newblock Recent developments in explicit constructions of extractors.
\newblock {\em Bulletin of the EATCS}, 77(67-95):10, 2002.

\bibitem[She83]{shen1983concept}
Alexander~Kh. Shen.
\newblock The concept of ($\alpha$, $\beta$)-stochasticity in the {K}olmogorov
  sense, and its properties.
\newblock {\em Soviet Math. Dokl.}, 28(1):295--299, 1983.

\bibitem[Smi07]{smith2007scrambling}
Adam~D. Smith.
\newblock Scrambling adversarial errors using few random bits, optimal
  information reconciliation, and better private codes.
\newblock {\em Symposium on Discrete Algorithms: Proceedings of the eighteenth
  annual ACM-SIAM symposium on Discrete algorithms}, 7(09):395--404, 2007.

\bibitem[Sol64]{sol:j:inductive}
Ray~J. Solomonoff.
\newblock A formal theory of inductive inference.
\newblock {\em Information and Control}, 7:224--254, 1964.

\bibitem[SUV17]{suv:b:kolmenglish}
Alexander Shen, Vladimir Uspensky, and Nikolay Vereshchagin.
\newblock {\em {K}olmogorov complexity and algorithmic randomness}.
\newblock American Mathematical Society, 2017.

\bibitem[Tya13]{tyagi2013common}
Himanshu Tyagi.
\newblock Common information and secret key capacity.
\newblock {\em IEEE Transactions on Information Theory}, 59(9):5627--5640,
  2013.

\bibitem[Vad12]{vad:b:pseudorand}
Salil~P. Vadhan.
\newblock Pseudorandomness.
\newblock {\em Foundations and Trends in Theoretical Computer Science},
  7(1-3):1--336, 2012.

\bibitem[Wyn75]{wyn:j:tap}
Aaron~D. Wyner.
\newblock The wire-tap channel.
\newblock {\em Bell Syst. Tech J.}, 54(8):1355--1387, 1975.

\bibitem[Zim17]{zim:c:kolmslepianwolf}
Marius Zimand.
\newblock {K}olmogorov complexity version of {S}lepian-{W}olf coding.
\newblock In {\em STOC 2017}, pages 22--32. ACM, June 2017.

\bibitem[ZL70]{zvo-lev:j:kol}
Alexander Zvonkin and Leonid Levin.
\newblock The complexity of finite objects and the development of the concepts
  of information and randomness by means of the theory of algorithms.
\newblock {\em Russian Mathematical Surveys}, 25(6):83--124, 1970.

\end{thebibliography}
\bibliographystyle{alpha}

\end{document}